\documentclass[10pt]{article}
\usepackage{amsmath}
\usepackage{amssymb}
\usepackage{amsthm}
\usepackage{epsfig}
\usepackage{euscript}

\usepackage{amsfonts}

\newtheorem{thm}{Theorem}[section]
\newtheorem{conj}[thm]{Conjecture}
\newtheorem{prp}[thm]{Proposition}
\newtheorem{lem}[thm]{Lemma}
\newtheorem{dfn}[thm]{Definition}
\newtheorem{cor}[thm]{Corollary}
\newtheorem{eg}[thm]{Example}
\newtheorem{rk}[thm]{Remark}
\newcommand{\ndnt}{\hspace{7mm}}
\newcommand{\hs}{\hspace{2mm}}

\newcommand{\s}{\vspace{3mm}}
\newcommand{\as}{\vspace{5mm}}

\newcommand{\C}{\mathbb{C}}
\newcommand{\N}{\mathbb{N}}

\newcommand{\R}{\mathbb{R}}
\newcommand{\Z}{\mathbb{Z}}
\newcommand{\bfi}{\textbf{i}}
\newcommand{\bfk}{\textbf{k}}
\newcommand{\bfp}{\textbf{p}}
\newcommand{\bfr}{\textbf{r}}
\newcommand{\bfq}{\textbf{q}}

\numberwithin{equation}{section} 

\begin{document}




{\Large
\begin{center}
{\sc partition zeta functions,\\
multifractal spectra, and\\
tapestries of complex dimensions}
\end{center}
\par}

\s

\begin{center}
Kate~E.~Ellis, Michel~L.~Lapidus,\\ Michael~C.~Mackenzie, and John~A.~Rock
\end{center}

\s
\begin{abstract}
For a Borel measure and a sequence of partitions on the unit interval, we define a multifractal spectrum based on coarse H\"older regularity. Specifically, the coarse H\"older regularity values attained by a given measure and with respect to a sequence of partitions generate a sequence of lengths (or rather, scales) which in turn define certain Dirichlet series, called the {\it partition zeta functions}. The abscissae of convergence of these functions define a multifractal spectrum whose concave envelope is the (geometric) Hausdorff multifractal spectrum which follows from a certain type of Moran construction. We discuss at some length the important special case of self-similar measures associated with weighted iterated function systems and, in particular, certain multinomial measures. Moreover, our multifractal spectrum is shown to extend to a {\it tapestry of complex dimensions} for a specific case of atomic measures.
\end{abstract}


\s

\textsc{Kate E. Ellis}\\
{\tiny \textsc{Department of Mathematics, California State
University, Stanislaus,\\ Turlock, CA} 95382 \textsc{USA} \par}

\textit{E-mail address:} \textbf{kellis1@csustan.edu}

\s

\textsc{Michel L. Lapidus}\\
{\tiny \textsc{Department of Mathematics, University of
California,\\
Riverside, CA} 92521-0135 \textsc{USA} \par}

\textit{E-mail address:} \textbf{lapidus@math.ucr.edu}

\s

\textsc{Michael C. Mackenzie}\\
{\tiny \textsc{Department of Mathematics, California State
University, Stanislaus,\\ Turlock, CA} 95382 \textsc{USA} \par}

\textit{E-mail address:} \textbf{michael.mackenzie@uconn.edu}

\s

\textsc{John A. Rock}\\
{\tiny \textsc{Department of Mathematics, California State
University, Stanislaus,\\ Turlock, CA} 95382 \textsc{USA} \par}

\textit{E-mail address:} \textbf{jrock@csustan.edu}

\s

\begin{itemize}
\item Subject class [2010]: Primary: 11M41, 28A12, 28A80. Secondary: 28A75, 28A78, 28C15
\item Keywords: Fractal string, geometric zeta function, Minkowski dimension, complex dimensions, oscillations, multifractal measure, multifractal spectrum, partition zeta function, tapestry of complex dimensions, regularity values, Cantor set, Hausdorff dimension, Moran construction, Moran set, Besicovitch subset, weighted iterated function systems, self-similar sets, self-similar and multinomial measures.
\item Acknowledgements: The work of the second author (M.~L.~Lapidus) was partially supported by the National Science Foundation under the research grant DMS--0707524.
\end{itemize}

\pagebreak

\tableofcontents

\section{Introduction}\label{intro}

Multifractal analysis is the study of physical, mathematical, dynamical, probabilistic, and statistical concepts in which a whole range of fractals may arise from a single object. Such phenomena are often modeled by measures that have highly irregular concentrations of mass. These richly structured measures are called {\it multifractal measures}, or simply {\it multifractals}, and arise from situations such as, but certainly not limited to, rainfall distribution, turbulence, distribution of galaxies, spatial distribution of earthquakes, internet traffic modeling, and modeling of financial time series. See, for example, \cite{EvMan,Falc,Fed,BM,Man,PF,PeitJS,Sch}.

\ndnt One setting for multifractal analysis that very much pertains to this paper is provided by number theory, specifically the study of $N$-ary (or base-$N$) expansions of real numbers, where $N$ is an integer greater than 1. The set of numbers in the unit interval with $N$-ary expansions containing the digits $0,1,\ldots,N-1$ in proportions given by a probability vector with $N$ components generate a fractal set. For a fixed $N$, the collection of the various fractal sets constructed in this manner provides a multifractal decomposition of the unit interval. A tool used to study the structure of these fractal sets is the Hausdorff dimension and this tool plays an important role in our approach to multifractal analysis. In general, the collection of Hausdorff dimension values determines a {\it multifractal spectrum} which describes the multifractal decomposition of a set (or rather, of a mass distribution). See, for instance, \cite{Bes1,BMP,CM,DekLi,Egg,EvMan,Falc,Fed,GrMauWi,LOW,Ol2,Ol3,Ol4}.

\ndnt Our primary objective here is the determination of multifractal spectra as the abscissae of convergence values for specific collections of Dirichlet series. The abscissa of convergence of a Dirichlet series is analogous to the radius of convergence of a power series and plays an important role in this work and in the theory of complex dimensions of fractal strings. Our technique is motivated by the determination of the Minkowski dimension of fractal strings as abscissae of convergence of geometric zeta functions (which are Dirichlet series themselves, see \cite{LapvF1,LapvF4,LapvF6} and \S\ref{IFS}). This determination allows for the definition of the generalization of Minkowski dimension called {\it complex dimensions} which are used, among other things, in expressions for counting functions and volume formulas in the study of the oscillatory phenomena of fractal strings.

\ndnt In our setting, we take a measure supported on a subset of the unit interval and determine its one-parameter family of {\it partition zeta functions}, indexed by a countably infinite collection of coarse H\"older regularity values (which we call {\it regularity values}), and their abscissae of convergence. Regularity values are the exponents $t$ for which a measure behaves locally like $r^t$, for small $r$, where $r$ is a positive real number that determines scale. In particular, we show that these abscissae of convergence recover classical forms of the geometric and symbolic Hausdorff multifractal spectra in certain cases and Hausdorff dimensions of Besicovitch subsets of self-similar fractals in others.

\ndnt This work, along with \cite{LLVR,LapRo1,Rock} and \cite[\S 13.3]{LapvF6} (which is an exposition of some of the work in those references), marks the beginning of a new theory of complex dimensions for multifractals. In particular, this work greatly expands upon the results presented in \cite{LapRo1} where the partition zeta functions and abscissa of convergence function for a generalized binomial measure supported on the Cantor set are developed and analyzed. Moreover, this paper provides significant strides toward the long-term goal stated (in a different but analogous manner) in \cite{LapRo1} of developing a theory of oscillatory phenomena which are intrinsic to multifractal geometries. This theory would parallel that developed for fractal strings in \cite{LapvF1,LapvF4,LapvF6} but would involve a whole family of partition zeta functions indexed by regularity and their complex dimensions.

\ndnt Other works which examine (from a different perspective) multifractal measures similar to those examined in this paper are \cite{AP,Bes1,Bes2,BMP,CM,DekLi,Egg,EvMan,Falc,Fed,FengLau,LapvF5,Lau,LOW,MinYa,Ol2,Ol3,Ol4}. A variety of other techniques in multifractal analysis can be found in
\cite{EM,Ellis,GrMauWi,Ja1,Ja2,Ja3,JaMey,LapvF5,LVM,LVT,LVV,BM,Man,PF,PeitJS}.
For the theory of complex fractal dimensions, one should consult the works
of M.~L.~Lapidus and M.~van~Frankenhuijsen
\cite{LapvF7,LapvF1,LapvF3,LapvF4,LapvF5,LapvF6} and their extensions with
B.~M.~Hambly, H.~Herichi, J.~L\'evy~V\'ehel, H.~Lu, E.~P.~J.~Pearse, J.~A.~Rock, and S.~Winter, accordingly, in \cite{HL,HerLap1,HerLap2,Lap4,LLVR,LapLu1,LapLu2,LapLu3,LapLu4,LapPe1,LapPe3,LapPe2,LapPeWin,LapRo1,Rock}.

\ndnt The remainder of this paper is organized as follows:

\ndnt \S\ref{IFS} provides a review of the relevant aspects of the theory of fractal strings and complex dimensions. In particular, we recall a connection between the counting function of the lengths of a fractal string and its complex dimensions from \cite{LapvF1,LapvF4}, which motivates the definition of a suitable counting function in our multifractal setting given in \S\ref{Definitions}.

\ndnt \S\ref{IMA} provides a review of the relevant aspects of multifractal analysis, in particular weighted iterated function systems, coarse H\"older regularity $\alpha$, geometric and symbolic Hausdorff multifractal spectra, and Besicovitch subsets of Moran fractals. These notions will play an important role in our approach, particularly in \S\ref{pzfssm}.

\ndnt \S\ref{Definitions} contains the definitions of our main objects of study: {\it $\alpha$-lengths}, {\it partition zeta function}, {\it abscissa of convergence function}, {\it complex dimensions with parameter} $\alpha$,  {\it tapestry of complex of dimensions}, and {\it counting function of the $\alpha$-lengths}.

\ndnt \S\ref{pzfssm} develops our main results regarding the partition zeta functions and abscissa of convergence functions of certain self-similar measures defined by weighted iterated function systems. In particular, connections with some of the well-known results found in \cite{Bes1,CM,Egg,Fed} and a recovery of the multifractal spectrum of the binomial measure on the unit interval (as described for example in \cite{EvMan,Falc,Fed}) are presented.

\ndnt \S\ref{pzfam} develops the partition zeta functions, abscissa of convergence functions, tapestries of complex dimensions, and counting functions for a certain collection of multifractal atomic measures. In particular, exact explicit formulas for the counting functions of the associated $\alpha$-lengths are given in terms of the underlying complex dimensions with regularity $\alpha$. These examples are among the first steps toward a new theory of complex dimensions and oscillatory phenomena for multifractals.

\ndnt \S\ref{conclusion} closes the paper with a discussion of related works in progress and ideas for further development of multifractal analysis via zeta functions.  

\section{Fractal Strings}\label{IFS}

\begin{figure}
\epsfysize=6.5cm\epsfbox{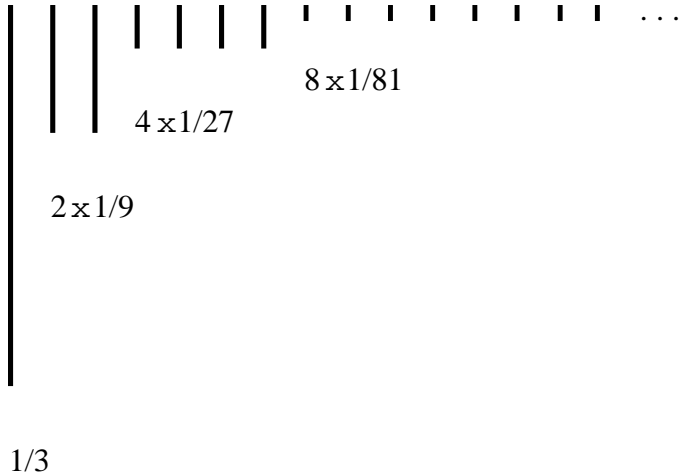}
    \caption{\textit{The lengths of the Cantor string.}}
\end{figure}

A brief review of fractal strings, geometric zeta functions and complex dimensions
(all of which are defined below) is in order. Results on fractal strings can be
found in \cite{HL,HeLap,Lap1,Lap2,Lap3,LLVR,LapLu1,LapLu2,LapLu3,LapLu4,LapPo1} and
results on geometric zeta functions and complex dimensions can be
found in \cite{HL,HerLap1,HerLap2,Lap4,LLVR,LapLu1,LapLu2,LapLu3,LapLu4,LapPe1,LapPe3,LapPe2,LapPeWin,LapRo1,LapvF7,LapvF1,LapvF3,LapvF4,LapvF5,LapvF6,Rock}.

\ndnt The primary references for the theory of complex dimensions of fractal strings are the monographs \cite{LapvF1} and \cite{LapvF4}. A significantly expanded second edition of \cite{LapvF4} is forthcoming in \cite{LapvF6}, but for the convenience of the reader, we will mostly refer to \cite{LapvF4} throughout the paper (except when required otherwise).

\subsection{Fractal Strings and Minkowski Dimension}\label{FSMD}

The primary example of a fractal string used throughout this work is the well-known Cantor string $\Omega_{CS}$ (the complement of the classical ternary Cantor set in the unit interval $[0,1]$). The first several lengths of the Cantor string appear in Fig.~1 and the Cantor string is discussed in much more detail in \S\ref{CS} below. 

\ndnt Fractal strings are defined as follows:

\begin{dfn}\label{dfn:fs}
A \underline{fractal string} $\Omega$ is a bounded open subset of
the real line.
\end{dfn}

\ndnt As in \cite{LLVR,LapvF1,LapvF4}, we distinguish a fractal string $\Omega$ from
its sequence of lengths $\mathcal{L}$ (with multiplicities). The sequence \( \mathcal{L}=\{\ell_i\}_{i=1}^{\infty}\) is the nonincreasing sequence of lengths of the disjoint open intervals \((a_i,b_i),\) where \(\Omega = \cup_{i=1}^{\infty}(a_i,b_i).\) More specifically, in this paper we follow the convention established in \cite{LapvF1,LapvF4} and generally do not consider the case where $\mathcal{L}$ comprises a finite collection of lengths. This is done in order to avoid discussion of trivial counterexamples to established results. We will, however, sometimes indicate what happens in the trivial case when $\mathcal{L}$ consists of finitely many lengths. Thus, throughout this paper, $\mathcal{L}$ typically comprises an infinite collection of lengths. We view $\mathcal{L}$  either as a decreasing sequence of positive \textit{distinct} lengths denoted $\{l_n\}_{n=1}^{\infty}$ along with their multiplicities $\{m_n\}_{n=1}^{\infty}$, or as a nonincreasing sequence of (possibly equal) positive lengths denoted $\{\ell_i\}_{i=1}^{\infty}$ and repeated according to their multiplicities.

\ndnt A generalization of Minkowski dimension called \textit{complex
dimension} (defined in \S\ref{CDCF} below) is used to study the properties of certain fractal
subsets of $\mathbb{R}$. For instance, the boundary of a fractal
string $\Omega$, denoted $\partial\Omega$, is often fractal and can be studied using
complex dimensions. Throughout this text, a fractal string
$\Omega$ is taken to be an open subset of the unit interval $[0,1]$
with $\mathcal{L}$ as its associated sequence of lengths.

\ndnt The volume of the (inner) tubular neighborhood of radius
$\varepsilon$ of the boundary $\partial\Omega$ of a fractal string $\Omega$ is
\[
V(\varepsilon)=|\{x \in \Omega \mid
dist(x,\partial\Omega)<\varepsilon\}|,
\]
where $|\cdot|$ denotes the Lebesgue measure (length). The
Minkowski dimension of $\partial\Omega$, or simply of
$\mathcal{L}$, is\footnote{This fractal dimension is also often called ``box dimension'' in the applied literature.} 
\[
\dim_M(\partial\Omega)=D:=\inf \{\alpha \geq 0 \mid
\limsup_{\varepsilon \rightarrow
0^{+}}V(\varepsilon)\varepsilon^{\alpha-1} <\infty \}.
\] Note that one may refer directly to the Minkowski dimension of the sequence of lengths $\mathcal{L}$ because $V(\varepsilon)$ can be shown to only depend on $\mathcal{L}$ (see \cite{LapPo1,LapvF4}).

\ndnt In \cite{Lap1}, it is shown that if $F = \partial\Omega$ is the boundary of a bounded open set
$\Omega$, then \(d-1 \leq \dim_H(F) \leq \dim_M(F) \leq d\) where
$d$ is the Euclidean dimension of the ambient space, $\dim_H(F)$ is
the Hausdorff dimension of $F$ and $\dim_M(F)$ is the Minkowski
dimension of $F$. We consider the case $d=1$ in this paper, thus
\[
0 \leq
\dim_H(F) \leq \dim_M(F) \leq 1.
\]
If $F$ is self-similar and further satisfies the Open Set Condition (defined in \S\ref{IMA}), it is well known that $\dim_H(F)=\dim_M(F)$. (See, e.g., \cite{Mor,Hut} and \cite[Ch. 9]{Falc}.)

\subsection{Complex Dimensions and Counting Functions}\label{CDCF}

The following equalities describe a relationship between the
Minkowski dimension of a fractal string $\Omega$ (taken to be the
Minkowski dimension of $\partial\Omega$) and the sum of each of its
lengths with exponent $\gamma \in \mathbb{R}$. This relationship with suitably defined Dirichlet series (later called {\it geometric zeta functions} in \cite{LapvF7,LapvF1,LapvF4}) was first
observed in \cite{Lap2} using a key result of A.~S.~Besicovitch and S.~J.~Taylor
\cite{BesTa}, and a direct proof can be found in \cite[pp.~17--18]{LapvF4}
(and can also be found in \cite{LapvF6}). We have\footnote{Strictly speaking, we must assume that $\mathcal{L}$ consists of infinitely many nonzero lengths; otherwise, $\dim_M(\partial\Omega)=\max\{0,D_{\mathcal{L}}\}$.}

\begin{equation}\label{eq:minkdimabsconv}
\dim_M(\partial\Omega)=D=D_{\mathcal{L}}:=\inf \left\{\gamma \in
\mathbb{R} \mid \sum_{i=1}^{\infty}\ell_{i}^{\gamma} <\infty
\right\}.
\end{equation}

Here, $D_{\mathcal{L}}$ can be considered to be the abscissa of
convergence of the Dirichlet series
\(\sum_{i=1}^{\infty}\ell_{i}^{s}\), where $s \in \mathbb{C}$. This
Dirichlet series is the \textit{geometric zeta function} of
$\mathcal{L}$; it is the function that has been generalized in
\cite{LLVR,LapRo1,Rock} using notions from multifractal analysis
and will in part motivate our proposed approach to multifractal zeta
functions.

\begin{dfn}\label{def:gzf} The \underline{geometric zeta function}
of a fractal string $\Omega$ with lengths $\mathcal{L}$ is
\begin{equation}\label{eq:gzf}
\zeta_{\mathcal{L}}(s):=\sum_{i=1}^{\infty}\ell_{i}^{s}=
\sum_{n=1}^{\infty}m_{n}l_{n}^{s}, 
\end{equation}
where $\textnormal{Re}(s)>D_{\mathcal{L}}$.
\end{dfn}
To consider lengths $\ell_i=0$, the convention $0^s=0$ for all $s
\in \mathbb{C}$ is used.

\ndnt One can extend the notion of the Minkowksi dimension of a fractal string
$\Omega$ to complex values by considering the poles of a meromorphic extension of
$\zeta_{\mathcal{L}}$. In general, $\zeta_{\mathcal{L}}$ may not
have a meromorphic extension to all of $\C$, yet one may consider suitable closed
regions $W \subset \C$ where $\zeta_{\mathcal{L}}$ has a meromorphic extension, and
collect the corresponding poles in these regions.

\ndnt Assume that $\zeta_{\mathcal{L}}$ has a meromorphic extension to a connected
open neighborhood of $W$ and there is no pole of
$\zeta_{\mathcal{L}}$ on $\partial W$. By a slight abuse of notation, $\zeta_{\mathcal{L}}$ denotes the geometric zeta function and its (necessarily unique) meromorphic extension to $W$.

\begin{rk}\label{rk:window}
\textnormal{More specifically, in \cite{LapvF4,LapvF6}, the `window' $W$ is chosen to be the closed subset of $\C$ that is to the right of the `screen' $S=\partial W$, defined as the graph (with the $x$ and $y$ axes interchanged) of a bounded and Lipschitz real-valued function on $(-\infty,D_{\mathcal{L}}]$; see \cite[\S 5.3]{LapvF4}.}
\end{rk}

\begin{dfn}\label{def:cd}
The set of \textnormal{(}visible\textnormal{)} \underline{complex dimensions} of a fractal string
$\Omega$ with lengths $\mathcal{L}$ is
\begin{equation}\label{eq:cd}
\mathcal{D}_{\mathcal{L}}(W):=\{\omega \in W \mid
\zeta_{\mathcal{L}} \textnormal{ has a pole at } \omega\}. 
\end{equation}
Furthermore, if $W=\C$, then $\mathcal{D}_{\mathcal{L}}:=\mathcal{D}_{\mathcal{L}}(\C)$ is simply called the set of complex dimensions of $\mathcal{L}$.
\end{dfn}

\ndnt The following proposition, which is a special case of \cite[Thm. 5.10]{LapvF4} (also found in \cite{LapvF6}), uses the complex dimensions $\mathcal{D}_{\mathcal{L}}(W)$ of a fractal string in a formula for the {\it geometric counting function} of $\mathcal{L}$, denoted $N_{\mathcal{L}}(x)$ and defined by
\[
N_{\mathcal{L}}(x) := \#\{ i \geq 1 \mid \ell_i^{-1} \leq x\} = \sum_{n \geq 1 \mid l_n^{-1} \leq x} m_n,
\]
where, as above, $\{\ell_i\}_{i=1}^{\infty}$ denotes the nonincreasing sequence of lengths of $\mathcal{L}$ repeated according to their multiplicities, whereas $\{l_n\}_{n=1}^{\infty}$ denotes the decreasing sequence of distinct lengths of $\mathcal{L}$ with associated multiplicities given by $\{m_n\}_{n=1}^{\infty}$.

\begin{prp}\label{prp:countingfunction}
Let $\Omega$ be a fractal string with lengths $\mathcal{L}$ such that $\mathcal{D}_{\mathcal{L}}(W)$ consists entirely of simple poles. Then, under certain mild growth conditions on $\zeta_{\mathcal{L}}$ \textnormal{(}namely, if $\zeta_{\mathcal{L}}$ is languid of a suitable order, in the sense of \textnormal{\cite{LapvF4,LapvF6}}\textnormal{)}, we have
\begin{equation}\label{eq:gcf}
N_{\mathcal{L}}(x) = \sum_{\omega \in \mathcal{D}_{\mathcal{L}}(W)}
\frac{x^{\omega}}{\omega}\textnormal{res}(\zeta_{\mathcal{L}}(s);\omega) + \{\zeta_{\mathcal{L}}(0)\} + R(x), 
\end{equation}
where $R(x)$ is an error term of small order and the term in braces is included only if $0 \in W\backslash \mathcal{D}_{\mathcal{L}}(W)$.
\end{prp}

\begin{rk}\label{rk:simple}
\textnormal{It is not necessary for the poles to be simple, but then the explicit formula for $N_{\mathcal{L}}$ is slightly more complicated to state; see \cite[\S 6.2.1]{LapvF4} for details.}
\end{rk}

\begin{rk}\label{rk:stronglylanguid}
\textnormal{If a fractal string $\Omega$ is {\it strongly languid}, then according to Theorem 5.14 of \cite{LapvF4}, Eq.~\eqref{eq:gcf} holds with no error term (i.e., $R(x) \equiv 0$). Examples of strongly languid fractal strings are self-similar strings (see \cite[Chs. 2 \& 3]{LapvF4}). In particular, the results for the counting functions to be presented in \S\ref{pzfam} follow from Theorem 5.14 of \cite{LapvF4}.}
\end{rk}

\ndnt Before continuing to the next section on multifractal analysis, consider the following results on classic examples of fractal strings---the Cantor string and the Fibonacci string. The Cantor string plays an important role throughout this paper and the Fibonacci string is recovered in Example \ref{eg:scottroby} below.

\subsection{The Cantor String and the Fibonacci String}\label{CS}

The {\it Cantor string} is defined as the open set $\Omega_{CS}$ that is
composed of all the deleted middle-third open intervals in the usual
construction of the classical ternary Cantor set. Hence, its
boundary $\partial\Omega_{CS}$ is simply the ternary Cantor set itself.
An approximation of the Cantor string appears in Fig.~2.

\begin{figure}
\epsfysize=1.25cm\epsfbox{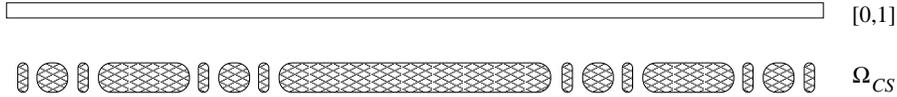}
    \caption{\textit{Approximation of the Cantor string $\Omega_{CS}$.}}
\end{figure}

\ndnt The distinct lengths of the Cantor string are given by $l_n = 3^{-n}$ with multiplicity $m_n = 2^{n-1}$ for every positive integer $n$; see Fig.~1. Hence, for \(\textnormal{Re}(s) > \log_3{2},\)
\begin{equation}\label{eq:gzfcs}
\zeta_{\mathcal{L}}(s)= \zeta_{CS}(s) =
\sum_{n=1}^{\infty}2^{n-1}3^{-ns} = \frac{3^{-s}}{1-2 \cdot 3^{-s}}. 
\end{equation}
Upon meromorphic continuation, we see that the last equality above
holds for all $s \in \mathbb{C}$; hence, for $j:=\sqrt{-1}$,
\begin{equation}\label{eq:cdcs}
\mathcal{D}_{\mathcal{L}} = \mathcal{D}_{CS} = \left\{ \log_3{2} +
\frac{2j z\pi}{\log3}  \mid  z \in \mathbb{Z} \right\} 
\end{equation}
and these poles are simple.

\ndnt In order to illustrate Proposition \ref{prp:countingfunction} above,
we give an exact formula for the counting function of the Cantor string, $N_{CS}$,
in terms of the complex dimensions $\mathcal{D}_{CS}$ (see \cite[Eq.~(1.31),~p.22]{LapvF4}). For all $x>1$, we have
\begin{equation*}\label{eq:gcfcs}
N_{CS}(x) = \frac{1}{2\log3} \sum_{z=-\infty}^{\infty}
\frac{x^{D+j zp}}{D+j zp}-1, 
\end{equation*}
where $D= \log_{3}2$ is the Minkowski dimension of the Cantor string (technically, of the Cantor set) and $p= 2\pi /\log3$ is its oscillatory period. This formula for $N_{CS}(x)$ is a special case of the explicit formula for the geometric counting function of general fractal strings provided by Theorems 5.10 and 5.14 in \cite{LapvF4}. Note that in light of this formula, $N_{CS}+1$ can be written as the product of $x^D$ and a multiplicatively periodic (or `log-periodic') function of $x$.

\ndnt Another example of a fractal string which is relevant to this paper is the Fibonacci string. The geometric zeta function of the Fibonacci string $\zeta_{\textnormal{Fib}}(s)$ is recovered as a special case in Example \ref{eg:scottroby}. (See \S 2.3.2 of \cite{LapvF4} for the development of the Fibonacci string.) The sequence of lengths for the Fibonacci string is
\begin{equation*}\label{eq:lfib}
\mathcal{L}_{\textnormal{Fib}}= \left\{2^{-n} \mid 2^{-n} \hs \textnormal{has multiplicity } F_{n+1}, n \in \N   \right\}, 
\end{equation*}
where $\N$ is the set of nonnegative integers and $F_n$ denotes the $n$th Fibonacci number. (Recall that $F_n$ is defined by the recurrence relation: $F_{n+1}=F_n+F_{n-1}$, and $F_0=0, F_1=1$.)

\ndnt Furthermore, the geometric zeta function $\zeta_{\textnormal{Fib}}$ of the Fibonacci string $\mathcal{L}_{\textnormal{Fib}}$ is given by
\begin{equation}\label{eq:gzffib}
\zeta_{\textnormal{Fib}}(s)=\sum_{n=0}^{\infty}F_{n+1}2^{-ns}=\frac{1}{1-2^{-s}-4^{-s}}. 
\end{equation}
The complex dimensions $\mathcal{D}_{\textnormal{Fib}}$ can be found by solving the quadratic equation 
\begin{equation*}\label{eq:solvefib}
2^{-2\omega}+2^{-\omega}=1, \ndnt \omega \in \C. 
\end{equation*}
Thus,
\begin{equation}\label{eq:cdfib}
\mathcal{D}_{\textnormal{Fib}}=\left\{D+jzp \mid z \in \Z \right\} \cup \left\{-D+j(z+1/2)p \mid z \in \Z \right\}, 
\end{equation}
where $\phi=(1+\sqrt{5})/2$ is the Golden Ratio, $D=\log_2{\phi}$, and $p = 2\pi/\log{2}$. In \S 2.3.2 of \cite{LapvF4}, these complex dimensions are used to determine a volume formula for the tubular neighborhood of the Fibonacci string.

\s

\ndnt We close \S\ref{CS} by noting that both the Cantor string and the Fibonacci string are {\it self-similar} fractal strings, in the sense of \cite[Chs.~2 \& 3]{LapvF4}.

\ndnt In the following section, we introduce some of the relevant theory and results on multifractal analysis currently available in the literature.

\section{Multifractal Analysis}\label{IMA}

Approaches to multifractal analysis which use multifractal measures closely related to those considered in this paper, and recalled below accordingly, can be found in \cite{AP,BarDem,Bes1,Bes2,BMP,CM,DekLi,Egg,EvMan,Falc,FengLau,LapvF5,Lau,LOW,MinYa,Ol2,Ol3,Ol4}. A variety of other techniques in multifractal analysis (incorporating wavelets, for instance) can be found in \cite{EM,Ellis,Ja1,Ja2,Ja3,JaMey,LapvF5,LVM,LVT,LVV,BM,Man,PF,PeitJS}.\footnote{The perspective adopted in all of these references is quite different, however, from the one adopted here, which consists in working with suitably defined partition (or multifractal) zeta functions.}  A common setting for multifractal analysis is that provided by self-similar measures defined by a probability vector and an Iterated Function System (IFS, or ``map specified Moran construction'' as in \cite{CM}) which satisfies the Open Set Condition (OSC, see \cite{Mor,Hut} and \cite{Falc}). We construct such measures in this section by following the development found in \cite[\S 1]{Ol4} and investigate the multifractal structure of these and other types of measures throughout this paper.

\subsection{Iterated Function Systems and Self-Similar Measures}\label{IFSSSM}

Multifractal analysis of measures is the study of the fractal structure of the sets $E_t$ of points $x\in E$ for which the measure $\mu(B(x,r))$ of the closed ball $B(x,r)$ with center $x$ and radius $r$ satisfies
\[
\lim_{r \rightarrow 0^+}\frac{\log{\mu(B(x,r))}}{\log{r}} = t,
\]
where $t\geq0$ is {\it local H\"older regularity} and $E$ is the support of $\mu$. That is, from this traditional perspective, multifractal analysis is the study of the fractal geometry of the sets $E_t$ where a Borel measure $\mu$ behaves locally like $r^t$.

\ndnt The setting for multifractal analysis provided by a self-similar measure uniquely defined by an IFS which satisfies the OSC and a probability vector is developed as follows. For positive integers $N$ and $d$ and each $i \in \{1,\dots,N\}$, let $S_i:\R^d \rightarrow \R^d$ be a contracting similarity with scaling ratio (or Lipschitz constant) $r_i \in (0,1)$. Let $\bfr = (r_1,\dots,r_N)$, and let $\bfp = (p_1,\ldots,p_N)$ be a probability vector. The collection of contracting similarities $\{S_i\}_{i=1}^{N}$ is said to satisfy the OSC if there exists a nonempty, bounded, and open set $V \subset \R^d$ such that $S_i(V) \subset V$ and $S_i(V) \cap S_k(V) = \emptyset$ for all $i \neq k$ with $i,k\in \{1,\ldots,N\}$. We note that in this paper, as with many others on multifractal analysis, the collection of functions $\{S_i\}_{i=1}^N$ is assumed to satisfy the OSC. However, \cite{Ol4} and \cite{STZ}, for instance, do not require the OSC to be satisfied.

\ndnt The multifractal measures for our setting are constructed as follows. Define the set $E$ and the self-similar measure $\mu$ (supported on $E$) to be the unique nonempty compact subset of $\R^d$ and the unique Borel probability measure which satisfy, respectively, $E = \bigcup_{i=1}^{N} S_i(E)$ and $\mu = \sum_{i=1}^{N}  p_i \mu \circ S_i^{-1}$ (see \cite{Hut}). In particular, the measures considered in Proposition \ref{prp:geomsym} and \S\ref{pzfssm} below are defined in this manner. The Hausdorff dimension of the the support $E$ is given by the solution of the Moran equation (see \cite{Mor})
\begin{equation}
\sum_{i=1}^N r_i^s = 1 \ndnt (s>0). \label{eq:moran}
\end{equation}

\subsection{Multifractal Spectra}\label{MS}

The multifractal spectra of Definitions \ref{dfn:ghms} and \ref{dfn:shms} along with Proposition \ref{prp:geomsym} below are presented as found in \cite{Ol4}, as well as the corresponding references therein. See especially the work of R.~Cawley and R.~D.~Mauldin in \cite{CM}.

\begin{dfn}\label{dfn:ghms}
The \underline{geometric Hausdorff multifractal spectrum} $f_g$ of a Borel measure $\mu$ \textnormal{(}on a Borel measurable subset of $\R^d$\textnormal{)} supported on $E$ is given by
\begin{equation}
f_g(t):=\dim_H(E_t), \label{eq:ghms}
\end{equation}
where $t\geq0$, $\dim_H$ is the Hausdorff dimension, and
\begin{equation}
E_t:=\left\{x \in E \mid \lim_{r \rightarrow 0^+}\frac{\log{\mu(B(x,r))}}{\log{r}} =t\right\}. \label{eq:ghmssubset}
\end{equation}
\end{dfn}

\ndnt The geometric Hausdorff multifractal spectrum $f_g$ is difficult to compute for general self-similar measures. Thus, the symbolic multifractal spectrum $f_s$ defined in terms of symbolic dynamics are often considered instead. This symbolic multifractal spectrum $f_s$, defined below, also serves as an analog to the approach to multifractal analysis developed in this paper.

\ndnt For a nonnegative integer $n$ and an integer $N \geq 2$, let
\[
\Lambda^* := \left\{\bfi = i_1\ldots i_n \mid k \in \N^*, i_k \in \left\{1,\ldots,N\right\}\right\}
\]
and
\[
\Lambda^{\N} := \left\{\bfi = i_1i_2\ldots \mid k \in \N^*, i_k \in \left\{1,\ldots,N\right\}\right\},
\]
where $\N = \{0,1,2,\ldots\}$. For $\bfi \in \Lambda^{\N}$, let $\bfi|n = i_1\ldots i_n$ be the truncation of $\bfi$ at the $n$th term. For $\bfi=i_1\ldots i_n \in \Lambda^*$, we let $S_{\bfi} := S_{i_1} \circ \ldots \circ S_{i_n}$ and $E_{\bfi} := S_{\bfi}(E)$. Likewise, let $p_{\bfi} := p_{i_1} \cdots p_{i_n}$ and $r_{\bfi} := r_{i_1} \cdots r_{i_n}$. Finally, let
$\pi : \Lambda^{\N} \rightarrow \R^d$ be defined by $\{\pi(\bfi)\} := \cap_{n=0}^{\infty} E_{\bfi|n}$.

\begin{dfn}\label{dfn:shms}
The \underline{symbolic Hausdorff multifractal spectrum} $f_s$ of a self-similar measure $\mu$ constructed as above is given by
\begin{equation}
f_s(t):=\dim_H\left(\pi\left\{\textnormal{\bfi} \in \Lambda^{\N}\mid \lim_{n \rightarrow \infty}\frac{\log{p_{\textnormal{\bfi}|n}}}{\log{r_{\textnormal{\bfi}|n}}} =t\right\}\right) \label{eq:shms}
\end{equation}
for $t\geq0$. \textnormal{(}Here and henceforth, given $A \subset \R^d$, $\dim_H(A)$ denotes the Hausdorff dimension of $A$.\textnormal{)}
\end{dfn}

\ndnt The function $f_s(t)$ is usually easier to analyze than $f_g(t)$. Define $b: \R \rightarrow \R$ by
\begin{equation}\label{eq:defb}
\sum_{i=1}^N p_i^q r_i^{b(q)} =1,
\end{equation}
and let $b^*: \R \rightarrow \R \cup \{-\infty\}$ be the Legendre transform of $b$. That is, for $t \in \R$,
\begin{equation}
b^*(t) := \inf_{q \in \R} (tq+b(q)).
\end{equation}

\ndnt When the OSC is satisfied, we have the following proposition (see, for instance, \cite{AP,CM,Ol4}).

\begin{prp}\label{prp:geomsym}
Let $\mu$ be the unique self-similar measure on $\R^d$ defined, as above, by the IFS associated with $\{S_i\}_{i=1}^{N}$ which satisfies the Open Set Condition and weighted by the probability vector $\textnormal{\bfp}$. Then
\begin{equation}\label{eq:ghmsshms}
f_g(-b'(q)) = f_s(-b'(q)) = b^*(-b'(q)), 
\end{equation}
where $q \in \R$ and $b'$ is the derivative of $b$ \textnormal{(}assumed to exist here\textnormal{)}. Moreover, for the support $E$ of the measure $\mu$, we have
\begin{equation}\label{eq:dimsupp}
\dim_H(E) = b(0) = b^*(-b'(0)). 
\end{equation}
\end{prp}

\ndnt Other well-known properties of the function $f=f_g=f_s$ are described in the following section.

\subsection{Properties of the Multifractal Spectrum}\label{Props}

\begin{figure}
\epsfysize=6.5cm\epsfbox{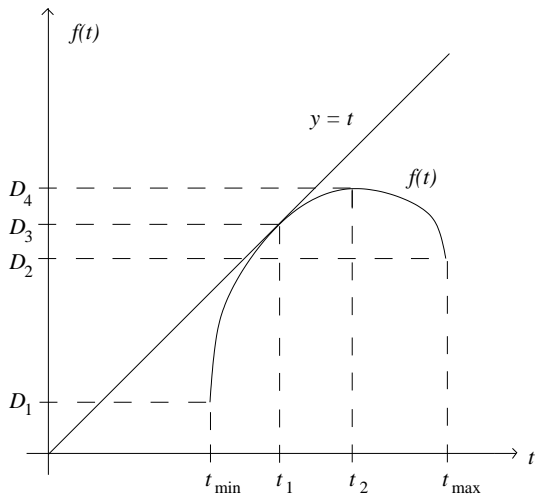}\label{cmmultispec}
    \caption{\textit{The graph of the multifractal spectrum $f(t)=f_g(t)=f_s(t)$ of a self-similar measure $\mu$ supported on a set defined by an IFS, where $t$ is local H\"older regularity. The behavior of $f$ is discussed in \S\ref{Props} and, for a specific case, is recovered in the context of the abscissa of convergence function and coarse H\"older regularity in \S\ref{RRCR}.}}
\end{figure}

A full development of the properties of the multifractal spectrum $f=f_g=f_s$  of a self-similar measure $\mu$ described in this section, some of which are displayed in Fig.~3, can be found in \cite[\S 1]{CM}, specifically Fig.~1.3 therein.

\ndnt A self-similar measure $\mu$ uniquely defined by $\{S_i\}_{i=1}^{N}$ and $\textnormal{\bfp}$ has maximum and minimum regularity values $t_{\min}$ and $t_{\max}$ which define the domain of $f$. These values are given by
\[
t_{\min} = \min_i\left\{\log_{r_i}p_i \mid i \in \{1,\ldots,N\} \right\}
\]
and
\[
t_{\max} = \max_i\left\{\log_{r_i}p_i \mid i \in \{1,\ldots,N\} \right\};
\]
hence the domain of $f$ (i.e., the values of $t$ for which $E_t$ is nonempty) is $[t_{\min},t_{\max}]$. See Fig.~3.

\ndnt Regardless of the values of $D_1=f(t_{\min})$ and $D_2=f(t_{\max})$, the slopes of $f$ at the points $(t_{\min},D_1)$ and $(t_{\max},D_2)$ are infinite. The value $D_3 = t_1 = f(t_1)$ is the {\it information dimension} of $\mu$. The value $D_4 = f(t_2) = \max\{f(t) \mid t \in [t_{\min},t_{\max}]\}$ is the Hausdorff dimension of $\mu$. Thus, by Proposition \ref{prp:geomsym} we have
\[
D_4 = f(t_2) = \max\{f(t) \mid t \in [t_{\min},t_{\max}]\} = \dim_H(E) = b^*(-b'(0)).
\]
See Fig.~3.

\ndnt If $p_i = r_i^D$ for all $i \in \{1,\dots,N\}$, then $D=D_4=f(D)$ and the domain of $f$ is the singleton $\{D\}$. Excluding this case, $f$ is concave and satisfies the following inequalities:
\[
f(t_{\min}) = D_1 < f(t_1) = D_3 < f(t_2) = D_4
\]
and
\[
f(t_{\max}) = D_2 < D_4.
\]

\subsection{Besicovitch Subsets of Moran Fractals}\label{BSMF}

Another setting of multifractal analysis that is important for our purposes is the one provided by the {\it Besicovitch subsets} of self-similar Moran fractals. (See \cite[\S 1]{MinYa}, for instance.)

\ndnt For an IFS which satisfies the OSC on the unit interval $[0,1]$ with scaling ratios $\bfr=(r_1,\ldots,r_N)$ and a probability vector $\bfq=(q_1,\ldots,q_N)$ (that is, $\sum_{i=1}^{N}q_i=1$ and $q_i\geq 0$ for $i \in \{1,\ldots,N\}$), the Besicovitch subset $E(\bfq)$ of the self-similar set $E$ (defined by the IFS) is defined by the coding mapping $\tau$ from $\{1,\ldots,N\}$ to $E$ as follows:

\begin{equation}\label{eq:bessubset}
E(\bfq) := \left\{ \tau(x) \in E \mid \lim_{n\rightarrow \infty}\frac{1}{n} \sum_{k=1}^n \chi_i(x_k) = q_i, 
x \in \{1,\ldots,N\}^{\mathbb{N}^*} \right\}, 
\end{equation}
where $i \in \{1,\ldots,N\}$, $\chi_i$ is the characteristic function of the singleton $\{i\}$, $x=(x_k)_{k=1}^{\infty}$ with $x_k\in\{1,\ldots,N\}$, and $\mathbb{N}^*$ is the set of positive integers.

\ndnt The well-known result in Proposition \ref{prp:hausbes} below follows from the results of A.~S.~Besicovitch in \cite{Bes2} and ties the results of Theorem \ref{thm:distinctreg} to existing theory. (See, also, \cite{CM,LOW,MinYa}.)

\begin{prp}\label{prp:hausbes}
For an IFS which satisfies the OSC on the unit interval $[0,1]$ with scaling ratios $\textnormal{\bfr}$ and a probability vector $\textnormal{\bfq}$, we have
\begin{equation}\label{eq:besprop}
\dim_H(E(\textnormal{\bfq})) = \frac{\sum_{i=1}^N q_i \log{q_i}}{\sum_{i=1}^N q_i \log{r_i}}.
\end{equation}
\end{prp}

\begin{rk}\label{rk:besegg}
\textnormal{In 1934, A.~S.~Besicovitch studied the unique nonterminating binary expansion of $x \in [0,1]$ (i.e., the case where $N=2$ and $x_k \in \{0,1\}$) in \cite{Bes1} and proved that}
\begin{equation*}\label{eq:besbin}
\dim_H(E(\textnormal{\bfq})) = \frac{-q_1\log{q_1} - q_2\log{q_2}}{\log{2}}.
\end{equation*}
\textnormal{In 1949, H.~G.~Eggleston generalized this result to $N$-ary expansions in \cite{Egg} and found that}
\begin{equation*}\label{eq:egg}
\dim_H(E(\textnormal{\bfq})) = \frac{-\sum_{i=1}^N q_i\log{q_i}}{\log{N}}.
\end{equation*}
\textnormal{(Also see \cite[Ch.~6]{Fed}.) In \S\ref{pzfssm}, we recover these results when $q_i \in \mathbb{Q}\cap[0,1]$ for $i \in \{1,\dots,N\}$.}
\end{rk}

\subsection{Coarse H\"older Regularity}\label{CHR}

As a break from the mold defined by the results in this section up to this point, the multifractal structure investigated in upcoming sections of this paper is based on the notion of {\it coarse H\"older regularity} as defined in \cite{LVT,LLVR,LapRo1,LapvF6,Rock}, for example, and simply called {\it regularity} in this paper. Regularity is key to the development of the partition zeta functions defined in the next section.

\begin{dfn}\label{def:reg}
For a given Borel measure $\mu$ with range in $[0,\infty]$ and an interval $U \subset [0,1]$ with positive Lebesgue measure \textnormal{(}denoted $|U|$\textnormal{)}, the \underline{regularity} \(A(U)\) of $U$ is
\begin{equation}\label{eq:reg}
A(U):=\frac{\log \mu(U)}{\log{|U|}}.
\end{equation}
\end{dfn}
Equivalently, $A(U)$ is the exponent $\alpha$ that satisfies
\[
|U|^{\alpha}=\mu(U).
\]
Note that regularity can be considered for any interval, whether
open, closed, or neither. To construct the partition zeta functions,
intervals are gathered according to their regularity.


\ndnt In general, regularity values $\alpha$ in the extended
real numbers $[-\infty,\infty]$ may be considered. For infinite regularity values, we take
\[
\alpha=\infty=A(U) \Leftrightarrow \mu(U) =0 \textnormal{ and } | U| > 0,
\]
and
\[
\alpha=-\infty=A(U) \Leftrightarrow \mu(U) = \infty  \textnormal{ and } | U | > 0.
\]
However, in \S\ref{pzfssm}, we only consider finite regularity values and in \S\ref{pzfam}, we consider $\alpha=\infty$ only briefly.

\ndnt Fixing the regularity $\alpha$ when taking a measure $\mu$ and a sequence of partitions $\mathfrak{P}$ into
consideration allows one to define the partition zeta functions, which is done in the next section.

\section{Definitions}\label{Definitions}

In this section, we define the main objects of study, in particular, the {\it partition zeta functions}. We consider self-similar measures which are supported on a subset of the unit interval $[0,1]$ and define the partition zeta functions in a manner which is similar to the way the geometric zeta functions are defined. (Compare Definition \ref{def:gzf}.) These definitions are similar in that both use a sequence of lengths to provide terms for certain Dirichlet series. However, unlike the geometric zeta functions, the lengths which define a partition zeta function are defined by the scales which stem from a sequence $\mathfrak{P}$ of partitions $\mathcal{P}_n$ (for $n \in \N^*$) of the unit interval and a fixed regularity value $\alpha$ (as defined in \S\ref{CHR}). To elaborate, we first define the appropriate sequence of lengths. Here and henceforth, we assume each partition $\mathcal{P}_n$ comprises a finite number of disjoint intervals with positive length.

\begin{dfn}\label{dfn:alphalengths}
For a Borel measure $\mu$ on the interval $[0,1]$ and a sequence of partitions $\mathfrak{P} = \{\mathcal{P}_n\}_{n=1}^{\infty}$ of $[0,1]$ with mesh tending to zero, the sequence of \underline{$\alpha$-lengths} $\mathcal{L}_{\mathfrak{P}}^{\mu}(\alpha)$ corresponding to regularity $\alpha \in [-\infty,\infty]$ is given by
\begin{equation}\label{eq:alphal}
\mathcal{L}_{\mathfrak{P}}^{\mu}(\alpha) :=
\left\{ \ell \mid \ell = |P^i_n| \textnormal{ and } A(P^i_n) = \alpha, \textnormal{ where } n \in \N^*, P^i_n \in \mathcal{P}_n \right\}.
\end{equation}
\end{dfn}

The $\alpha$-lengths are essentially the distinct scales, along with their multiplicities, of the intervals at level $n$ (for all $n \in \N^*$) in the partition $\mathcal{P}_n \in \mathfrak{P}$ which have regularity $\alpha$.  In turn, the $\alpha$-lengths are used to define the partition zeta function as follows:

\begin{dfn}\label{def:pzf}
For a Borel measure $\mu$ on the interval $[0,1]$ and a sequence of partitions $\mathfrak{P} = \{\mathcal{P}_n\}_{n=1}^{\infty}$ with mesh tending to zero, the \underline{partition zeta function} $\zeta^{\mu}_{\mathfrak{P}}(\alpha,s)$ corresponding to regularity $\alpha \in [-\infty,\infty]$ is given by
\begin{equation}\label{eq:pzfl}
\zeta^{\mu}_{\mathfrak{P}}(\alpha,s):=\zeta_{\mathcal{L}}(s),
\end{equation}
where $\mathcal{L}:=\mathcal{L}_{\mathfrak{P}}^{\mu}(\alpha)$ and $\textnormal{Re}(s)$ is large enough. That is,
\begin{equation}\label{eq:pzf}
\zeta^{\mu}_{\mathfrak{P}}(\alpha,s)
:=\sum_{\ell \in \mathcal{L}^{\mu}_{\mathfrak{P}}(\alpha)} \ell^s
=\sum_{n=1}^{\infty}\sum_{A(P^i_n)=\alpha}|P^{i}_{n}|^s,
\end{equation}
where the inner sum is taken over the intervals $P^i_n$ with regularity $A(P^i_n)=\alpha$ in the
partition $\mathcal{P}_n$ for each $n \in \N^*$, and $\textnormal{Re}(s)>D_{\mathcal{L}}$, with $\mathcal{L}:=\mathcal{L}_{\mathfrak{P}}^{\mu}(\alpha)$ as in Definition \ref{dfn:alphalengths} above.
\end{dfn}

\ndnt If there is no interval $P^i_n$ such that $A(P^i_n)=\alpha_0$ for some regularity value $\alpha_0$, then we set $\zeta^{\mu}_{\mathfrak{P}}(\alpha_0,\cdot)$ identically equal to zero and we refer to such regularity values as {\it trivial regularity values}. By extension, we will call {\it trivial} the regularity values $\alpha$ for which $\zeta^{\mu}_{\mathfrak{P}}(\alpha,\cdot)$ is an entire function. This is the case, for example, if $\mathcal{L}_{\mathfrak{P}}^{\mu}(\alpha)$ consists of at most finitely many nonzero $\alpha$-lengths and is true iff we are in that situation provided $\zeta^{\mu}_{\mathfrak{P}}(\alpha,s)$ is assumed to have a meromorphic continuation to an open neigborhood of $[\textnormal{Re}(s)\geq 0]$; see Remark \ref{rk:finitelengths} below.

\begin{rk}\label{rk:countablereg}
\textnormal{Unlike the case of the multifractal zeta functions from \cite{Gibson,LLVR,LapvF6,Rock}, there is no particular partition zeta function which directly corresponds to the lengths of the fractal string of the complement in $[0,1]$ of the support of a given measure. Still, there is a connection to the Hausdorff (and Minkowski) dimension of the support, as will be seen in \S\ref{pzfssm}. Also, for a given measure $\mu$, there are at most countably many nontrivial regularity values attained with respect to a given sequence of partitions $\mathfrak{P}$. Hence, for a given measure and sequence of partitions, there are at most countably many partition zeta functions that are not entire (possibly in the broader context of Remark \ref{rk:developtheory} below).}
\end{rk}

\ndnt In order to connect our methods to those outlined in \S\ref{IMA}, with motivation provided by Definition \ref{def:gzf} and Proposition \ref{prp:geomsym}, we consider the function $f_{\mathfrak{P}}^{\mu}$ on $[-\infty,\infty]$ which maps the regularity values $\alpha$ attained by a measure $\mu$ with respect to a sequence of partitions $\mathfrak{P}$ to the abscissa of convergence of the corresponding partition zeta function $\zeta^{\mu}_{\mathfrak{P}}(\alpha,\cdot).$

\begin{dfn}\label{dfn:abscon}
Given a Borel measure $\mu$ on $[0,1]$ and a sequence of partitions $\mathfrak{P}$ with mesh tending to zero, the \underline{abscissa of convergence function} $f_{\mathfrak{P}}^{\mu}(\alpha)$ is given by
\begin{equation}\label{eq:acf}
f_{\mathfrak{P}}^{\mu}(\alpha) := \inf \left\{\gamma \in \mathbb{R} \mid
\zeta^{\mu}_{\mathfrak{P}}(\alpha,\gamma) <\infty \right\},
\end{equation}
for $\alpha \in [-\infty,\infty]$. For a trivial regularity value $\alpha_0$, we set $f_{\mathfrak{P}}^{\mu}(\alpha_0)=0$. 
\end{dfn}

That is, more precisely, in general $f_{\mathfrak{P}}^{\mu}(\alpha)$ is defined as the maximum of $0$ and the abscissa of convergence \textnormal{(}given by Eq.~\eqref{eq:acf}\textnormal{)} of the Dirichlet series defining $\zeta^{\mu}_{\mathfrak{P}}(\alpha,\cdot)$; so that $f_{\mathfrak{P}}^{\mu}(\alpha) \geq 0$ for all nontrivial regularity values $\alpha \in \R$ and when $\zeta^{\mu}_{\mathfrak{P}}(\alpha_0,\cdot)$ is entire for a given value of $\alpha_0$ \textnormal{(}i.e., when $\alpha_0$ is trivial\textnormal{)}, $f_{\mathfrak{P}}^{\mu}(\alpha_0)=\max\{0,-\infty\}=0$. \textnormal{(}See Remark \ref{rk:finitelengths}.\textnormal{)}  Accordingly, for a nontrivial regularity value $\alpha$, $\{s \in \C \mid \textnormal{Re}(s) > f_{\mathfrak{P}}^{\mu}(\alpha)\}$ is the largest open right half-plane on which the Dirichlet series in Eq.~\eqref{eq:pzf} is absolutely convergent.

\begin{rk}\label{rk:finitelengths}
\textnormal{Note that if $\mathcal{L}_{\mathfrak{P}}^{\mu}(\alpha)$ consists of infinitely many nonzero $\alpha$-lengths and $\zeta^{\mu}_{\mathfrak{P}}(\alpha,s)$ admits a meromorphic continuation to an open neigborhood of $[\textnormal{Re}(s)\geq 0]$, then by Eq.~\eqref{eq:pzf}, $\zeta^{\mu}_{\mathfrak{P}}(\alpha,0)=\sum_{\ell\in\mathcal{L}_{\mathfrak{P}}^{\mu}(\alpha)}1=\infty$. Hence, by definition, the abscissa of convergence of $\zeta^{\mu}_{\mathfrak{P}}(\alpha,\cdot)$ is necessarily nonnegative in this case; in particular, the partition zeta function $\zeta^{\mu}_{\mathfrak{P}}(\alpha,\cdot)$ cannot be entire. Conversely, if $\mathcal{L}_{\mathfrak{P}}^{\mu}(\alpha)$ consists of finitely many nonzero $\alpha$-lengths, then $\zeta^{\mu}_{\mathfrak{P}}(\alpha,\cdot)$ is clearly entire and hence, its abscissa of convergence is $-\infty$; from which it follows that $f_{\mathfrak{P}}^{\mu}(\alpha)=\max\{0,-\infty\}=0$. This justifies, in particular, our terminology for trivial/nontrivial regularity values and the more precise definition of $f_{\mathfrak{P}}^{\mu}(\alpha)$ given in the text immediately following Definition \ref{dfn:abscon}. That is, regularity $\alpha$ is {\it trivial} iff $\mathcal{L}_{\mathfrak{P}}^{\mu}(\alpha)$ consists of finitely many nonzero lengths, or more generally, if $\zeta^{\mu}_{\mathfrak{P}}(\alpha,\cdot)$ is entire. Otherwise, $\alpha$ is {\it nontrivial}.}
\end{rk}

\ndnt When the partition zeta function $\zeta^{\mu}_{\mathfrak{P}}(\alpha,\cdot)$ has a meromorphic continuation to a window $W_{\alpha} \subset \C$, we have the following definitions. Note that with a mild abuse of notation, $\zeta^{\mu}_{\mathfrak{P}}(\alpha,s)$ denotes the partition zeta function as well as its meromorphic continuation to $W_{\alpha}$.

\begin{dfn}\label{dfn:poles} For a Borel measure $\mu$ on $[0,1]$, sequence
$\mathfrak{P}$ of partitions of $[0,1]$ with mesh tending to zero, and regularity $\alpha \in [-\infty,\infty]$, the set of\\
\underline{complex dimensions with parameter $\alpha$}, denoted by
$\mathcal{D}^{\mu}_{\mathfrak{P}}(\alpha, W_{\alpha})$, is given by
\begin{equation}\label{eq:cdalpha}
\mathcal{D}^{\mu}_{\mathfrak{P}}(\alpha, W_{\alpha}) := \{\omega \in W_{\alpha} \mid \zeta^{\mu}_{\mathcal{N}}(\alpha,s) \textnormal{ has a pole at } \omega \},
\end{equation}
for an appropriate window $W_{\alpha}$.
\end{dfn}

Gathering the sets $\mathcal{D}^{\mu}_{\mathfrak{P}}(\alpha, W_{\alpha})$ by all finite regularity values $\alpha$ attained by $\mu$ with respect to $\mathfrak{P}$ yields the following collection.
 
\begin{dfn}\label{dfn:tapestry}
For a Borel measure $\mu$ on $[0,1]$ and sequence $\mathfrak{P}$ of partitions of $[0,1]$ with mesh tending to zero,
the \underline{tapestry of complex dimensions} $\mathcal{T}_{\mathfrak{P}}^{\mu}$ with respect to the windows $W_{\alpha}$ is given by
\begin{equation}\label{eq:tap}
\mathcal{T}_{\mathfrak{P}}^{\mu} := \left\{ (\alpha,\omega) \mid \alpha \in (-\infty,\infty), \omega \in \mathcal{D}^{\mu}_{\mathfrak{P}}(\alpha, W_{\alpha}) \right\}.
\end{equation}
\end{dfn}

\ndnt Note that by definition, we have
\[
\mathcal{T}_{\mathfrak{P}}^{\mu} \subset \R \times W_{\alpha} \subset \R \times \C.
\]

\ndnt In light of Definition \ref{dfn:poles} and Proposition \ref{prp:countingfunction}, we define as follows the {\it counting function of the $\alpha$-lengths} of a measure $\mu$ with respect to a sequence of partitions $\mathfrak{P}$ and the attained regularity values $\alpha$.

\begin{dfn}
For $\mathcal{L} := \mathcal{L}_{\mathfrak{P}}^{\mu}(\alpha)$, the \underline{counting function of the $\alpha$-lengths} of $\mu$ with respect to $\mathfrak{P}$ is
\begin{equation}\label{eq:calphal}
N_{\mathfrak{P}}^{\mu}(\alpha, x) := N_{\mathcal{L}}(x),
\end{equation}
for $x>0$.
\end{dfn}

\begin{rk}
\textnormal{Note that $N_{\mathfrak{P}}^{\mu}(\alpha, x)$ does not correspond to a fractal string per se, just a sequence of lengths (or rather, scales). In this setting, the corresponding explicit formulas for general fractal strings are discussed in \cite[Ch. 5]{LapvF4} (also \cite{LapvF6}), and as mentioned in \S\ref{CS}, can immediately be used to describe the multiscale behavior of $\mu$ with respect to $\mathfrak{P}$ and regularity $\alpha$.}
\end{rk}

\ndnt The following section investigates the properties of the partition zeta functions and the abscissa of convergence functions for self-similar measures and their natural sequences of partitions.

\section{Partition Zeta Functions of\\ Self-Similar Measures}\label{pzfssm}

We now develop and investigate the partition zeta functions of a self-similar measure $\mu$ on $[0,1]$ and its natural sequence of partitions of $[0,1]$, denoted $\mathfrak{P}$, where both $\mu$ and $\mathfrak{P}$ are defined by an IFS and a probability vector as in \S\ref{IFSSSM}. The scaling ratios given by $\bfr$, along with the probability vector $\bfp$, completely determine the regularity values $\alpha$ which the measure $\mu$ attains with respect to $\mathfrak{P}$.

\subsection{Self-Similar Measures and\\ Natural Sequences of Partitions}\label{SSMNSP}

Consider an IFS with contracting similarities $\{S_i\}_{i=1}^{N}$ which satisfy the OSC and have scaling ratios $\bfr=(r_1,\ldots,r_N)$ such that $\sum_{i=1}^{N}r_i\leq1$ and $r_i>0$ for each $i \in \{1,\ldots,N\}$. (Here and thereafter, we have $N\geq 2$.) Without loss of generality, we may assume that $0 \in S_1([0,1])$ and $1 \in S_N([0,1])$. Furthermore, let $\bfp = (p_1,\ldots,p_N)$ be a probability vector; that is, $\sum_{i=1}^{N}p_i=1$ and, without loss of generality in the setting of this section, $p_i>0$ for all $i \in \{1,\ldots,N\}$.

\ndnt At each stage $K \in \N^* = \{1,2,\ldots\}$ of the recursive construction of the IFS, there are $N^K$ distinct closed intervals with positive mass. These intervals, along with the distinct intervals which fill in the gaps between them (if any), constitute the partitions $\mathcal{P}_K$ for each $K \in \N^*$. In turn, $\{\mathcal{P}_K\}_{K=1}^{\infty}$ constitutes the {\it natural sequence of partitions} $\mathfrak{P}$. Each of the intervals in $\mathcal{P}_K$ with positive mass have their length and mass completely described by an ordered $N$-tuple of nonnegative integers $\bfk = (k_1,\ldots,k_N)$ such that $\sum_{i=1}^{N}k_i=K$. That is, an interval $P \in \mathcal{P}_K$ with positive mass has length $|P|$ of the form $r_1^{k_1}\cdots r_N^{k_N}$ and mass $\mu(P)$ of the form $p_1^{k_1}\cdots p_N^{k_N}$. See Fig.~4, where we use the notation introduced below.

\ndnt Ultimately, $\bfk = (k_1,\ldots,k_N)$ and $\bfr = (r_1,\ldots,r_N)$ define the regularity value $\alpha(\bfk)$ as below and, in turn, the Besicovitch subset $E(\bfk/K)$, defined by Eq.~\eqref{eq:bessubset}, of the Moran (self-similar) fractal $E = \textnormal{supp}(\mu)$. See \cite{CM,LOW,MinYa,Ol4} for the construction of Moran fractals and related results in more general settings.

\begin{rk} \textnormal{Intervals without mass have regularity $\alpha=\infty$, but we do not investigate this case (except in \S\ref{pzfam}, and there only briefly) since the resulting partition zeta functions are divergent everywhere. This is in stark contrast to the results obtained under the context of the multifractal zeta functions as described in \cite{Gibson,LLVR,LapRo1,LapvF6,Rock}, where regularity $\alpha=\infty$ precisely recovers the geometric zeta function and $\alpha=-\infty$ yields the Hausdorff dimension of the boundary of a certain type of fractal string.}
\end{rk}

\ndnt The breakdown of mass and length as above provides a complete description of the regularity values attained by a self-similar measure $\mu$ with respect to the natural sequence of partitions $\mathfrak{P}$. Specifically, $\bfk = (k_1,...,k_N)$ corresponds to the regularity value $\alpha=\alpha(\bfk)$ given by
\[
\alpha(\bfk)=
\frac{\log{(p_1^{k_1}\cdots p_N^{k_N})}}{\log{(r_1^{k_1}\cdots r_N^{k_N})}}=
\frac{\sum_{i=1}^N k_i \log{p_i}}{\sum_{i=1}^N k_i \log{r_i}},
\]
where the convention $0\log0=0$ is used. Note the similarity to the ratio of logarithms used in Eq.~\eqref{eq:shms}. Before stating and deriving our main results, we consider the partition zeta functions and abscissa of convergence function of a  well-known binomial measure.

\begin{figure}
\epsfysize=6.2cm\epsfbox{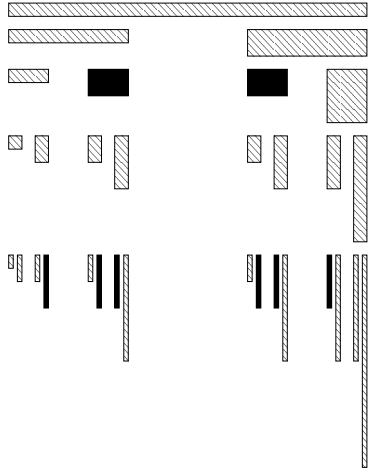}\label{binomcs}
    \caption{\textit{The solid black blocks correspond to intervals which have regularity $\alpha(1,2)$ stemming from the measure $\beta$ \textnormal{(}see \S\ref{MMCS}\textnormal{)} and its natural sequence of partitions.}}
\end{figure}

\subsection{A Multifractal Measure on the Cantor Set}\label{MMCS}

Consider a binomial measure $\beta$ supported on the classical ternary Cantor set defined by the weighted IFS given by
$\bfr=(1/3,1/3)$ with probabilities $\bfp=(p_1,p_2)$ such that $0 < p_1 < p_2<1$. At every stage $K \in \N^*$ of
the IFS, there are $2^K$ intervals with positive mass. The number of intervals at stage $K$ with regularity $\alpha(k_2,K):=\alpha(\bfk)$ is given by the binomial coefficient
\[
\binom{K}{k_1} = \binom{K}{k_2} = \frac{K!}{k_1!k_2!},
\]
where the first equality holds since $k_1=K-k_2$. See Fig.~4 for a depiction of the first five stages of the weighted IFS with $N=2$, $\bfr=(1/3,1/3)$, and $\bfp=(1/3,2/3)$, resulting in the self-similar measure $\beta$ and the natural sequence of partitions $\mathfrak{P}$.

\ndnt The partition zeta function $\zeta_{\mathfrak{P}}^{\beta}(\alpha,s)$ with regularity $\alpha = \alpha(k_2,K)$ is given by
\[
\zeta_{\mathfrak{P}}^{\beta}(\alpha,s) = \sum_{n=1}^{\infty}\binom{nK}{nk_2}3^{-nKs}.
\]
The abscissa of convergence function $f=f_{\mathfrak{P}}^{\beta}(\alpha)$ is given by
\begin{align*}
f_{\mathfrak{P}}^{\beta}(\alpha) =& -x\log_3x - (1-x)\log_3(1-x)\\
=&-\left(\frac{1-\alpha}{\log_3{2}}\right)\log_3{\left(\frac{1-\alpha}{\log_3{2}}\right)}\\
 &- \left(1-\frac{1-\alpha}{\log_3{2}}\right)\log_3{\left(1-\frac{1-\alpha}{\log_3{2}}\right)},
\end{align*}
where $x = k_2/K = \frac{1- \alpha}{\log_3{2}}.$ (See Fig.~5.) The development of the partition zeta functions
$\zeta_{\mathfrak{P}}^{\beta}(\alpha,s)$ and the abscissa of convergence function $f_{\mathfrak{P}}^{\beta}(\alpha)$ are provided below in more general settings. 

\begin{figure}
\epsfysize=6.2cm\epsfbox{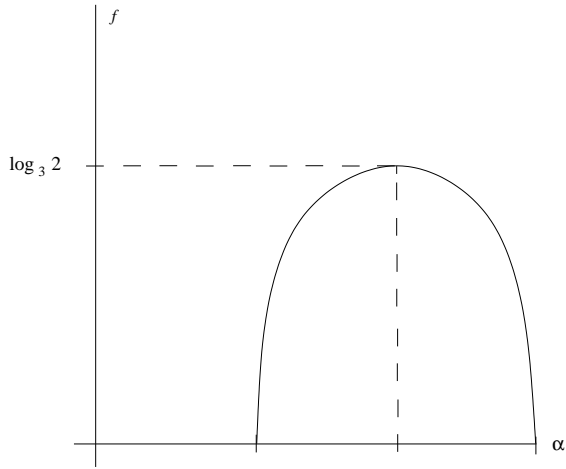}\label{betaspec}
    \caption{\textit{The graph of $f=f_{\mathfrak{P}}^{\beta}(\alpha)$ as a function of regularity $\alpha$ for the binomial measure $\beta$ supported on the Cantor set.}}
\end{figure}

\subsection{First Main Result: Distinct Regularity}\label{FMR}

For a self-similar measure defined by an IFS and a probability vector, the task of collecting the intervals with identical regularity values from every partition in $\mathfrak{P}$ is relatively simple when the conditions of Theorem \ref{thm:distinctreg} are satisfied.

\ndnt Here and henceforth, we let
\[
\binom{m}{m_1\ldots m_N} := \frac{m!}{m_1!\cdots m_N!}
\]
denote the multinomial coefficients, where $m_i \in \N$ for $i \in \{1,\ldots,N\}$ and $m := \sum_{i=1}^{N}m_i$.

\begin{thm}\label{thm:distinctreg}
Let $\mu$ be the self-similar measure and $\mathfrak{P}$ be the natural sequence of partitions defined by an IFS with scaling ratios $\textnormal{\bfr}=(r_1,\ldots,r_N)$ which satisfies the OSC and is weighted by a probability vector $\textnormal{\bfp}$ \textnormal{(}as described in \S\ref{SSMNSP}\textnormal{)}. Consider the following hypothesis:
\begin{itemize}
\item[\textnormal{(\textbf{H})}] Suppose that for all $\textnormal{\bfk} = (k_1,\dots,k_N) \in \N^N$ where  $\gcd(k_1,\dots,k_N)=1$ and $\textnormal{\bfk} \neq \textnormal{\textbf{0}}$, we have that the regularity values $\alpha(\textnormal{\bfk})$ are distinct. \textnormal{(}See Remark \ref{rk:distinctreg} below.\textnormal{)} That is, suppose \(\alpha(z_1,...,z_N) = \alpha(\textnormal{\bfk})\) if and only if there exists $m \in \N^*$ such that
    $z_i = mk_i$ for all $i \in \{1,\ldots,N\}$.
\end{itemize}
Assume that hypothesis \textnormal{(\textbf{H})} holds. Then, for any $\textnormal{\bfk} = (k_1,\dots,k_N) \in \N^N$  with \(\gcd(k_1,\dots,k_N)=1\), and letting \(K:=\sum_{i=1}^{N}k_i\), we have
\begin{equation}\label{eq:pzfdistinctreg}
\zeta^{\mu}_{\mathfrak{P}}(\alpha(\textnormal{\bfk}),s) =
\sum_{n=1}^{\infty}\binom{nK}{nk_1\cdots nk_N}(r_1^{k_1}\cdots
r_N^{k_N})^{ns}.
\end{equation}

\ndnt Moreover, the abscissa of convergence $\sigma=f_{\mathfrak{P}}^{\mu}(\alpha(\textnormal{\bfk}))$ of the partition zeta function $\zeta^{\mu}_{\mathfrak{P}}(\alpha(\textnormal{\bfk}),s)$ is given by
\begin{equation}\label{eq:msdistinctreg}
f_{\mathfrak{P}}^{\mu}(\alpha(\textnormal{\bfk})) = \frac{\sum_{i=1}^N (k_i/K)\log (k_i/K)}{\sum_{i=1}^N
(k_i/K)\log{r_i}} = \dim_H(E(\textnormal{\bfk}/K)),
\end{equation}
where $E(\textnormal{\bfk}/K)$ is the Besicovitch subset of the Moran (self-similar) fractal $E=\textnormal{supp}(\mu)$ defined by the scaling ratios $\bfr = (r_1,\ldots,r_N)$ and the probability vector $\textnormal{\bfk}/K=(k_1/K,\ldots,k_N/K)$; see Eq.~\eqref{eq:bessubset} in \S\ref{BSMF}. Equivalently, and with use of the convention $0^0=1$, the abscissa of convergence $\sigma$ is the unique real number satisfying the equation
\begin{equation}\label{eq:absconvsats}
(r_1^{k_1}\cdots r_N^{k_N})^{\sigma}\frac{K^{K}}{k_1^{k_1} \cdots k_N^{k_N}}=1;
\end{equation}
in addition, $\sigma>0$.
\end{thm}

\s

\begin{proof} Every interval from $\mathfrak{P}$ is taken into
account since \(\gcd(k_1,\dots,k_N)=1\). Specifically, for the given positive integers $(z_1,...,z_N)$, there is an integer $m \in \N^*$ such that $mk_i=z_i$ for each $i$; hence the corresponding intervals
have the same regularity since $\alpha(\bfk) = \alpha(n\bfk)$ for all $n \in \N^*$. Thus, $\alpha=\alpha(\bfk)$ is attained by $\mu$ in $\mathcal{P}_{nK}$ for each $n \in \N^*$ and the corresponding intervals contribute their lengths to the same partition zeta function. The coefficients $\binom{nK}{nk_1\cdots nk_N}$ stem from the multinomial distribution of mass among the intervals in the partitions $\mathcal{P}_{nK}$.

\ndnt To determine the abscissa of convergence function $f_{\mathfrak{P}}^{\mu}(\alpha(\bfk))$, an application of Stirling's formula and the $n$-th root test allows for the
formulation of the abscissa of convergence, written $\sigma$ for notational convenience, in terms of the
$N$-tuples  $\bfk=(k_1,\dots,k_N)$ and $\bfr=(r_1,\ldots,r_N)$. Indeed, for a fixed real number $s$, Stirling's formula yields
\begin{align*}
\kappa_{n,s} &:= \binom{nK}{nk_1\cdots nk_N}(r_1^{k_1}\cdots r_N^{k_N})^{ns}\\
&= \frac{(nK)!}{(nk_1)!\cdots (nk_N)!}(r_1^{k_1}\cdots r_N^{k_N})^{ns} \\
&= \frac{(r_1^{k_1}\cdots r_N^{k_N})^{ns} K^{nK}}{k_1^{nk_1} \cdots k_N^{nk_N}} \cdot \frac{\sqrt{K}}{\sqrt{k_1} \sqrt{2\pi nk_2} \cdots \sqrt{2\pi nk_N}}(1+\varepsilon_{n}),
\end{align*}
where $\varepsilon_{n} \rightarrow 0$ as $n \rightarrow \infty$. Hence,
\[
\kappa_{n,s}^{1/n}= \frac{(r_1^{k_1}\cdots r_N^{k_N})^{s} K^{K}}{k_1^{k_1} \cdots k_N^{k_N}}(1+\delta_{n}),
\]
where $\delta_{n} \rightarrow 0$ as $n \rightarrow \infty$. Therefore, according to the $n$-th root test, the numerical series
\[
\sum_{n=1}^{\infty} \kappa_{n,s} = \sum_{n=1}^{\infty}\binom{nK}{nk_1\cdots nk_N}(r_1^{k_1}\cdots
r_N^{k_N})^{ns}
\]
converges for $s > \rho$ and diverges for $s < \rho$, where $\rho$ is the unique real number such that
\[
1=(r_1^{k_1}\cdots r_N^{k_N})^{\rho}\frac{K^{K}}{k_1^{k_1} \cdots k_N^{k_N}}.
\]
Equivalently,
\[
f_{\mathfrak{P}}^{\mu}(\alpha(\bfk))= \rho = \log_{r_1^{k_1}\cdots r_N^{k_N}}\left(\frac{k_1^{k_1} \cdots k_N^{k_N}}{K^K}\right) = \frac{\sum_{i=1}^N (k_i/K)\log (k_i/K)}{\sum_{i=1}^N  (k_i/K)\log{r_i}}.
\]
(The fact that $\rho$ is well defined and $\rho >0$ will be explained at the end of the proof.) By definition of the abscissa of convergence $\sigma$ (see Definition \ref{dfn:abscon}), it follows that
\[
f_{\mathfrak{P}}^{\mu}(\alpha(\bfk))= \sigma = \rho,
\]
which in light of the previous expression for $\rho$, establishes part of Theorem \ref{thm:distinctreg}.

\ndnt Finally, by Proposition \ref{prp:hausbes} applied to $\bfr$ and the probability vector  $\bfq = \bfk/K$, we have
\[
f_{\mathfrak{P}}^{\mu}(\alpha(\textnormal{\bfk}))= \sigma= \frac{\sum_{i=1}^N (k_i/K)\log (k_i/K)}{\sum_{i=1}^N
(k_i/K)\log{r_i}} = \dim_H(E(\textnormal{\bfk}/K)),
\]
where $E(\bfk/K)$ is the Besicovitch subset referred to in Theorem \ref{thm:distinctreg}. (See, also, \cite{CM,LOW,MinYa}.) This last equality, combined with the fact that $\sigma=\rho$, enables us to conclude the proof of the main statement of Theorem \ref{thm:distinctreg}.

\ndnt As promised, we now supplement the proof by briefly explaining why $\rho$ is well defined as the unique real solution of the equation $\varphi(s)=1$, where
\[
\varphi(s):= (r_1^{k_1} \cdots r_N^{k_N})^s\frac{K^{K}}{k_1^{k_1} \cdots k_N^{k_N}}.
\]
First, note that $\varphi$ is strictly decreasing on $(-\infty,\infty)$; indeed, $\varphi'(s)<0$ since $0<r_i<1$ for $i=1,\ldots,N$. Second, observe that $\varphi(0)=K^{K}/(k_1^{k_1} \cdots k_N^{k_N})>1$
since  $K=\sum_{i=1}^{N}k_i$ and not all of the integers $k_i$ are zero for $i=1,\ldots,N$. (As noted above, the convention $0^0=1$ is used.) Hence, $\rho$ is well defined and $\rho >0$, as desired.
\end{proof}

\begin{rk}\label{rk:distinctreg}
\textnormal{We believe that the assumption (\textbf{H}) made in Theorem \ref{thm:distinctreg} is essentially superfluous. That is, we expect that the regularity values $\alpha(\bfk)$ are either always distinct or can be broken down into distinct values with corresponding multiplicities, as is done (in a special case) in Lemma \ref{lem:regssm} below.}
\end{rk}

\begin{rk}\label{rk:betadistinctreg}
\textnormal{Theorem \ref{thm:distinctreg} applies to the results on the binomial measure $\beta$ presented in \S\ref{MMCS}. In that setting, $\alpha(k_2,K) = \log_r(p_1)+(k_2/K)\log_r\left(p_2/p_1\right)$ and these regularity values are distinct as long as $p_1<p_2$. Indeed, using the substitution $x = k_2/K$, we have that $\alpha(k_2,K) = \alpha(x) = \log_r(p_1)+x\log_r\left(p_2/p_1\right)$ is a nonconstant linear function.}
\end{rk}

\ndnt The condition of Theorem \ref{thm:distinctreg} requiring \(\alpha(z_1,...,z_N) = \alpha(\textnormal{\bfk})\) if and only if there exists $m \in \N^*$ such that $z_i = mk_i$ for all $i \in \{1,\ldots,N\}$ is not a necessary condition. Indeed, the next section shows that this condition can be replaced with rational independence of the logarithm of the probability values when there is a single distinct scaling ratio used in the IFS. Moreover, the following corollary shows that, at least in the specific case of the binomial measure $\beta_0$ on the unit interval $[0,1]$, we have $\hat{f}_{\mathfrak{P}}^{\beta_0}(t) = f_g(t) = f_s(t) = b^*(t)$,
where $\hat{f}_{\mathfrak{P}}^{\beta_0}$ is the concave envelope of $f_{\mathfrak{P}}^{\beta_0}$ on the interval $[t_{\min},t_{\max}]$. (See \cite{Bes1,Bes2,DekLi,Egg,MinYa}, as well as Figs.~3 and 5, but note that for $\beta_0$ we have $\bfr=(1/2,1/2)$ and for $\beta$ in \S\ref{MMCS} we have $\bfr=(1/3,1/3)$.)

\ndnt The following corollary also stems from the discussion that follows \cite[Theorem 4.2]{LapRo1} and the connection to the binomial measure on the unit interval (called $\beta_0$ below) discussed, for instance, in \cite{EvMan,Falc,Fed}.

\begin{cor}\label{cor:recoverbinom}
Consider the binomial measure $\beta_0$ defined by the similarities $S_1(x)=x/2$ and $S_2(x)=x/2+1/2$ with scaling ratios $\textnormal{\bfr}= (1/2,1/2)$ and the probability vector $\textnormal{\bfp} = (1/3,2/3)$. Then, for all $t \in [t_{\min},t_{\max}]$, we have
\begin{equation}\label{eq:correcoverbinom}
\hat{f}_{\mathfrak{P}}^{\beta_0}(t) = f_g(t) = f_s(t) = b^*(t),
\end{equation}
where $\hat{f}_{\mathfrak{P}}^{\beta_0}$ is the concave envelope of $f_{\mathfrak{P}}^{\beta_0}$ on the interval $[t_{\min},t_{\max}]$. \textnormal{(}Here, $[t_{\min},t_{\max}]=[\log_{2}3-1,\log_{2}3]$.\textnormal{)}
\end{cor}

\begin{proof}
By Theorem \ref{thm:distinctreg}, we have
\begin{align*}
f_{\mathfrak{P}}^{\beta_0}(\alpha(\bfk)) &= \frac{(k_1/K)\log(k_1/K)+(k_2/K)\log(k_2/K)}{(k_1/K)\log{r_1}+(k_2/K)\log{r_2}}\\
&= -\frac{k_1}{K}\log_2\left(\frac{k_1}{K}\right)-\frac{k_2}{K}\log_2\left(\frac{k_2}{K}\right)\\
&= -\left(1-\frac{k_2}{K}\right)\log_2\left(1-\frac{k_2}{K}\right)-\frac{k_2}{K}\log_2\left(\frac{k_2}{K}\right).
\end{align*}
Further, we have
\[
\alpha(\bfk) = \frac{\log(2^{k_2}/3^K)}{\log(1/2^K)} = \log_2{3}-\frac{k_2}{K}.
\]
Since $k_2$ is a nonnegative integer and $K$ is a positive integer where $k_2 \leq K$, the maximum value of $\alpha(\bfk)$ is $\log_2{3}=t_{\max}$ and the minimum value of $\alpha(\bfk)$ is $\log_2{3}-1=t_{\min}$.

\ndnt Now, according to \cite{EvMan}, we have
\[
f_g(t) = -\frac{t_{\max}-t}{t_{\max}-t_{\min}}\log_2\left(\frac{t_{\max}-t}{t_{\max}-t_{\min}}\right)
         -\frac{t-t_{\min}}{t_{\max}-t_{\min}}\log_2\left(\frac{t-t_{\min}}{t_{\max}-t_{\min}}\right),
\]
where $t \in [t_{\min},t_{\max}] = [\log_2{3}-1,\log_2{3}]$.

\ndnt For $t_1=\log_2{3}-k_2/K=\alpha(\bfk)$, we have
$t_{\max}-t_{\min}=1,$ $t_{\max}-t_1=k_2/K$, and $t_1-t_{\min}=1-k_2/K$.
Substitution yields
\[
f_{\mathfrak{P}}^{\beta_0}(t_1) = f_g(t_1).
\]
\ndnt Note that the collection of all regularity values $t_1=\log_2{3}-k_2/K=\alpha(\bfk)$ is a dense subset of $[t_{\min},t_{\max}]=[\log_2{3}-1,\log_2{3}]$; furthermore, recall that $f_g(t)$ is concave (see \S\ref{Props}). Therefore, by Proposition \ref{prp:geomsym} we have
\[
\hat{f}_{\mathfrak{P}}^{\beta_0}(t) = f_g(t) = f_s(t) = b^*(t)
\]
on $[t_{\min},t_{\max}]$.
\end{proof}

\begin{rk}\label{rk:binom}
\textnormal{Corollary \ref{cor:recoverbinom} immediately holds in the slightly more general case where the components of the probability vector $\bfp$ are distinct (that is, $p_1\neq p_2$). However, in this paper we address only the case $\bfp=(1/3,2/3)$, for clarity of exposition.}
\end{rk}

\ndnt The next corollary is a simple consequence of Theorem \ref{thm:distinctreg} and the Moran equation.  

\begin{cor}\label{cor:recoverhausdorff}
Assume that the conditions of Theorem \ref{thm:distinctreg} hold and, additionally, that $\bfr=(\lambda,\ldots,\lambda)$ where $0<\lambda \leq 1/N$. Then 
\begin{equation}\label{eq:correcoverhaus}
f_{\mathfrak{P}}^{\mu}(\alpha(1,\ldots,1)) = \dim_H(\textnormal{supp}(\mu)).
\end{equation}
\end{cor}

\begin{proof}
In this case, $K=N$ and by Theorem \ref{thm:distinctreg} we have 
\[
f_{\mathfrak{P}}^{\mu}(\alpha(1,\ldots,1)) = -\log_{\lambda}N.
\]
The Moran equation (Eq.~\eqref{eq:moran}) then becomes 
\[
\sum_{i=1}^N \lambda^s = N \lambda^s =1,
\] 
which has the unique solution $s=-\log_{\lambda}N$ when $s>0$.
\end{proof}

\ndnt In light of Theorem \ref{thm:distinctreg}, Theorem \ref{thm:smr} (an analog of Theorem \ref{thm:distinctreg} in \S\ref{SMR} below),  Corollary \ref{cor:recoverbinom}, and Corollary \ref{cor:recoverhausdorff}, we make the following conjecture:

\begin{conj}\label{conj:spectra}
For a self-similar measure $\mu$ on $[0,1]$ and its natural sequence of partitions $\mathfrak{P}$, we have
\begin{equation}\label{eq:conj}
\hat{f}_{\mathfrak{P}}^{\mu}(t) = f_g(t) = f_s(t) = b^*(t),
\end{equation}
for all $t \in [t_{\min},t_{\max}]$, where $\hat{f}_{\mathfrak{P}}^{\mu}(t)$ is the concave envelope of $f_{\mathfrak{P}}^{\mu}(\alpha)$ on $[t_{\min},t_{\max}]$.
\end{conj}

\begin{rk}\label{rk:conj}
\textnormal{Conjecture \ref{conj:spectra} would be proven in the case where the regularity values $\alpha(\bfk)$ are distinct if, for instance, one could show that for all
$\bfk=(k_1,\dots,k_N) \in \N^N$ such that $\gcd(k_1,\dots,k_N)=1$,}
\begin{equation*}\label{eq:conjrk}
E_{\alpha(\textnormal{\bfk})} = E(\textnormal{\bfk}/K),
\end{equation*}
\textnormal{where $E_{\alpha(\textnormal{\bfk})}$ is given by Eq.~\eqref{eq:ghmssubset} with $t=\alpha(\textnormal{\bfk})$ and $E(\textnormal{\bfk}/K)$ is given by Eq.~\eqref{eq:bessubset} with $\textnormal{\bfq}=\textnormal{\bfk}/K$. See \cite[Thm. A]{LOW}, \cite[Prop. 5.1]{FengLau}, and \cite{CM} for further evidence for the validity of Conjecture \ref{conj:spectra}.}
\end{rk}

\ndnt One condition of Theorem \ref{thm:distinctreg} that is required throughout the rest of this paper is that the components of $\bfk = (k_1,\dots,k_N)$ satisfy $\gcd(k_1,\dots,k_N)=1$. This condition guarantees that every interval stemming from $\mathfrak{P}$ with regularity $\alpha(\bfk)$ is taken into account in the corresponding partition zeta function. The next section develops and analyzes such partition zeta functions for a specific class of self-similar measures.

\subsection{Second Main Result: Two Distinct Probabilities\\
and a Single Scaling Ratio}\label{SMR}

Consider an IFS with $N$ contracting similarities and a single scaling ratio (Lipschitz constant) $r$ such that $0<r=r_i \leq 1/N$, for all $i \in \{1,\ldots,N\}$. In contrast to what happens for the results in \S\ref{FMR}, the collection of finite regularity values attained by the corresponding measure $\mu$ on the natural sequence of partitions $\mathfrak{P}$ in this setting is determined solely by the probability vector $\bfp$.

\ndnt Suppose there are $w \geq 2$ $(w \in \N^*)$ distinct values among the components of $\bfp=(p_1,\ldots,p_N)$. (The case where $w=1$ is a special case of Example \ref{eg:mono} below.) Denote the distinct probabilities by $\bfp' = (p_1',\ldots,p_w')$, and denote the multiplicity of each of the numbers $p_i'$ by $c_i$ for $i = 1,\ldots,w$. In this setting, \(\sum_{q=1}^{w}c_q=N\) and \(p_1^{k_1}\ldots p_N^{k_N} = (p_1')^{k_1'}\ldots (p_w')^{k_w'},\) where $\sum_{i=1}^N k_i = \sum_{q=1}^w k_q' = K$. We continue to use the convention $0\log0=0$, and therefore, $0^0=1$.

\begin{lem}\label{lem:regssm}
With $\mu$ and $\mathfrak{P}$ as above, if $\gcd(k_1',\ldots,k_w')=1$ and the numbers $\log_r{p_1'},\ldots,\log_r{p_w'}$ are rationally independent, then the distinct regularity values attained by $\mu$ on $\mathfrak{P}$ are given by
\begin{equation}\label{eq:lemreg}
\alpha(k_1',\ldots,k_w') = \frac{1}{K} \log_r \left((p_1')^{k_1'}\ldots (p_w')^{k_w'}\right).
\end{equation}
Moreover, for every $n \in \N^*$, the number of intervals $P$ with regularity value
\begin{equation*}\label{eq:lemregn}
\alpha(nk_1',\ldots,nk_w')=\alpha(k_1',\ldots,k_w')
\end{equation*}
in the partition $\mathcal{P}_{nK}$ is
\begin{equation*}\label{eq:lemmult}
\binom{nK}{nk_1,\ldots,nk_N} = \binom{nK}{nk_1',\ldots,nk_w'}c_1^{nk_1'}\cdots c_w^{nk_w'}.
\end{equation*}
\end{lem}

\begin{proof}
For every $N$-tuple of nonnegative integers $\bfk$, there exists an $n \in \N$ such that
\[
\alpha(\bfk)=\alpha(nk_1',\ldots,nk_w').
\]
The rational independence of the numbers $\log_r{p_i'}$ and the fact that
\[
\log_{r^K}\left((p_1')^{k_1'}\ldots (p_w')^{k_w'}\right) = \frac{1}{K}\sum_{i=1}^w k_i'\log_r{p_i'}
\]
imply that the regularity values $\alpha(k_1',\ldots,k_w')$ are indeed distinct. Moreover, we have
\[
(p_1+\cdots+p_N)^{nK} = (c_1p_1'+\cdots+c_wp_w')^{nK},
\]
which immediately yields
\[
\binom{nK}{nk_1,\ldots,nk_N} = \binom{nK}{nk_1',\ldots,nk_w'}c_1^{nk_1'}\cdots c_w^{nk_w'}.
\]
\end{proof}

\ndnt Lemma \ref{lem:regssm} allows us to determine the corresponding partition zeta functions and abscissa of convergence function in Proposition \ref{prp:pzfssm}, as we now explain.

\begin{prp}\label{prp:pzfssm}
If the conditions of Lemma \ref{lem:regssm} are satisfied, then
\begin{equation}\label{eq:prppzfssm}
\zeta^{\mu}_{\mathfrak{P}}(\alpha(\textnormal{\bfk}),s) =
\sum_{n=1}^{\infty} \binom{nK}{nk_1',\ldots,nk_w'}c_1^{nk_1'}\cdots c_w^{nk_w'}r^{nKs}
\end{equation}
and
\begin{equation}\label{eq:prpacfssm}
f_{\mathfrak{P}}^{\mu}(\alpha(\textnormal{\bfk}))=f_{\mathfrak{P}}^{\mu}(\alpha(k_1,...,k_N)) =
\log_{r^{K}} \left( \frac{(k_1')^{k_1'}\cdots (k_w')^{k_w'}}{c_1^{k_1'}\cdots c_w^{k_w'}K^K} \right).
\end{equation}
\end{prp}

\begin{proof}
Eq.~\eqref{eq:prppzfssm} follows at once from Lemma \ref{lem:regssm}. Next, much as in the proof of Theorem \ref{thm:distinctreg}, it can be seen that an application of Stirling's formula and the $n$-th root test implies that the abscissa of convergence $\gamma:=f_{\mathfrak{P}}^{\mu}(\alpha(\textnormal{\bfk}))$ of $\zeta^{\mu}_{\mathfrak{P}}(\alpha(\textnormal{\bfk}),\cdot)$ satisfies
\[
1= (r^{K})^{\gamma}\frac{c_1^{k_1'}\cdots c_w^{k_w'}  K^{K}}{(k_1')^{k_1'}
    \cdots (k_w')^{k_w'}}.
\]
That is, $\gamma>0$ is the unique real solution of the above equation. Hence,
\[
f_{\mathfrak{P}}^{\mu}(\alpha(\bfk)) =
\gamma = \log_{r^{K}} \left( \frac{(k_1')^{k_1'}\cdots (k_w')^{k_w'}}{c_1^{k_1'}\cdots c_w^{k_w'}K^K} \right),
\]
as desired.
\end{proof}

\ndnt In the special case of the conditions specified in Lemma \ref{lem:regssm} where $r_1= \ldots=r_N=r$ and $w = 2$, the abscissa of convergence function $f_{\mathfrak{P}}^{\mu}(\alpha(\bfk))$ can easily be expressed as a function of the single variable $\alpha = \alpha(k_2',K) = \alpha(\bfk)$. The results then mirror those described in \cite{CM,Ol4}, among others, in the case of ``map specified'' Moran fractals (i.e., generated by an IFS) as popularized by \cite{BarDem}.

\begin{thm}\label{thm:smr}
Assuming the conditions of Lemma \ref{lem:regssm} are satisfied and, specifically, $w=2$ and $r_1=\ldots=r_N=r$, then
we have
\begin{equation}\label{eq:pzfthmsmr}
\zeta^{\mu}_{\mathfrak{P}}(\alpha(\textnormal{\bfk}),s) = \zeta^{\mu}_{\mathfrak{P}}(\alpha(k_2',K),s)
= \sum_{n=1}^{\infty} \binom{nK}{nk_2'}(N-c_2)^{n(K-k_2')} c_2^{nk_2'}r^{nKs}
\end{equation}
and
\begin{equation}\label{eq:acfthmsmr}
f_{\mathfrak{P}}^{\mu}(\alpha(\textnormal{\bfk})) =
f_{\mathfrak{P}}^{\mu}(\alpha(k_2',K)) =
\log_{r^K} \left( \frac{(K-k_2')^{K-k_2'} (k_2')^{k_2'}}{(N-c_2)^{K-k_2'}c_2^{k_2'}K^K} \right).
\end{equation}
\textnormal{(}Recall that $c_1$ and $c_2$ denote the multiplicities, respectively, of the two distinct values of $p_1,\ldots,p_N$.\textnormal{)}

\ndnt Moreover, the concave envelope $\hat{f}_{\mathfrak{P}}^{\mu}$ of $f_{\mathfrak{P}}^{\mu}$ has infinite slope at the extreme values of the attained regularity values.

\ndnt Lastly, we have
\begin{equation}\label{eq:maxthmsmr}
\max_{\alpha(\textnormal{\bfk})}\{f_{\mathfrak{P}}^{\mu}(\alpha(\textnormal{\bfk}))\} = f_{\mathfrak{P}}^{\mu}(\alpha(c_2/N)) = \dim_H(\textnormal{supp}(\mu)).
\end{equation}
\end{thm}

\begin{proof} We have
\begin{align*}
\alpha(k_2',K) &= \frac{ \log{ \left( (p_1')^{K-k_2'}(p_2')^{k_2'} \right)}}{ \log{ \left(r^{K}\right) }}\\
&=\log_r{(p_1')} + \frac{k_2'}{K}\log_r{\left(\frac{p_2'}{p_1'}\right)}.
\end{align*}

Letting $x = k_2'/K$, we have
\[
x=\frac{\alpha - \log_r(p_1')}{\log_r(p_2'/p_1')}.
\]
This substitution allows one to express the abscissa of convergence function $g(x):=f_{\mathfrak{P}}^{\mu}(\alpha(k_2',K))$ in the following form:
\begin{align*}
g(x)=f_{\mathfrak{P}}^{\mu}(\alpha(k_2',K))
&= \log_{r^K} \left( \frac{(K-k_2')^{K-k_2'}(k_2')^{k_2'}}{(N-c_2)^{K-k_2'}c_2^{k_2'}K^K} \right)\\
&= x \log_{r}\left(\frac{x}{c_2}\right) + (1-x)\log_{r}\left(\frac{1-x}{N-c_2}\right).
\end{align*}
\ndnt By temporarily allowing $x\in [0,1]$ and using a slight abuse of notation, we deduce that the first two derivatives of $g$ are given by
\[
g'(x) = \log_r\left(\frac{x}{1-x}\right)+\log_r\left(\frac{c_1}{c_2}\right),\ndnt
g''(x) = \frac{-1}{x(1-x)\log{r}}.
\]
It immediately follows, assuming without loss of generality that $p_1'<p_2'$ and $x \in (0,1)$, that
\begin{align*}
&\lim_{x\rightarrow 0^+}g(x) = -\log_r{c_1}, \hs \lim_{x\rightarrow 1^-}g(x)=-\log_r{c_2},\\
&\lim_{x\rightarrow 0^+}g'(x)= \infty, \hs \lim_{x\rightarrow 1^-}g'(x)=-\infty,
\end{align*}
and 
\[
g''(x)<0,
\] 
which implies that $\hat{f}^{\mu}_{\mathfrak{P}}$ is concave. A bit more calculus then shows that
\[
\max_{\alpha(\textnormal{\bfk})}\{f_{\mathfrak{P}}^{\mu}(\alpha(\textnormal{\bfk}))\} = f_{\mathfrak{P}}^{\mu}(\alpha((c_1,c_2))) = -\log_r{N}=\dim_H(\textnormal{supp}(\mu)).
\]
In particular, the last equality holds since $s=-\log_r{N}$ is the unique real-valued solution of the equation $Nr^s=1$. (See Proposition \ref{prp:geomsym} and Eq.~\eqref{eq:defb}.)
\end{proof}

\ndnt The following section discusses the way in which the results of \S\ref{SMR} recover recent as well as classical results on self-similar measures.

\subsection{Recovery of Recent and Classical Results}\label{RRCR}

Theorem \ref{thm:smr} allows for the recovery of recent and classical results from the multifractal analysis of self-similar measures in the context provided by the partition zeta functions and the abscissa of convergence function, as we now discuss.

\begin{figure}
\epsfysize=6.2cm\epsfbox{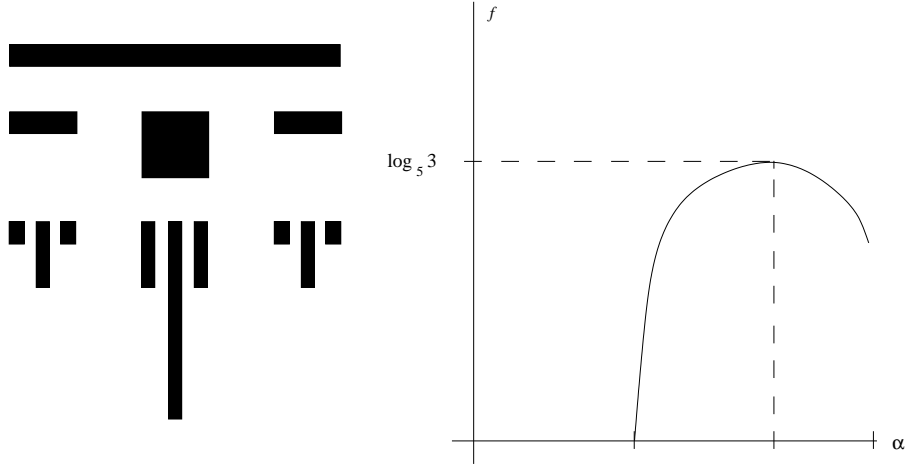}\label{psispec}
    \caption{\textit{The first few stages in the construction of the trident measure $\psi$ and its natural sequence of partitions $\mathfrak{P}$ \textnormal{(}left\textnormal{)}, along with the graph of its abscissa of convergence function $f=f_{\mathfrak{P}}^{\psi}(\alpha)$ as a function of regularity $\alpha$ \textnormal{(}right\textnormal{)}.}}
\end{figure}

\begin{eg}[The Trident Measure]\label{eg:trident}
\textnormal{If $N=3, \bfr=(1/5,1/5,1/5)$, and $\bfp=(1/5,3/5,1/5)$, then $w=2, c_1=2$ and $c_2=1$; the resulting self-similar measure is called the {\it trident measure} $\psi$ (see Fig.~6).
The distinct regularity values attained by $\psi$ on its natural sequence of partitions $\mathfrak{P}$ are}
\begin{equation*}\label{eq:regtrident}
\alpha(\textnormal{\bfk}) =
\frac{\log{(3^{nk_2'}/5^{nK})}}{\log{(1/5^{nK})}} = 1 -
\frac{k_2'}{K}\log_5{3};
\end{equation*}
\textnormal{furthermore, the partition zeta functions are} 	 												
\begin{equation*}\label{eq:pzftrident}
\zeta^{\psi}_{\mathfrak{P}}(\alpha(\textnormal{\bfk}),s) = \sum_{n=1}^{\infty}
\binom{nK}{nk_2'} 2^{n(K-k_2')} 5^{-nKs},
\end{equation*}
\textnormal{and the abscissa of convergence function is}	
\begin{equation*}\label{eq:acftrident}
f^{\psi}_{\mathfrak{P}}(\alpha(\textnormal{\bfk})) = -x\log_5{(x)}-(1-x)\log_5{(1-x)}+(1-x)\log_5{2},
\end{equation*}	
\textnormal{where $x = k_2'/K = \frac{1-\alpha}{\log_5{3}}$ and $\alpha(\bfk)=\alpha$. Thus,}
\begin{align*}
f^{\psi}_{\mathfrak{P}}(\alpha(\textnormal{\bfk})) =
f^{\psi}_{\mathfrak{P}}(\alpha) =& -\left(\frac{1-\alpha}{\log_5{3}}\right)\log_5{\left(\frac{1-\alpha}{\log_5{3}}\right)}\\
& - \left(1-\frac{1-\alpha}{\log_5{3}}\right)\log_5{\left(1-\frac{1-\alpha}{\log_5{3}}\right)}\\
& + \left(1-\frac{1-\alpha}{\log_5{3}}\right)\log_5{2}.
\end{align*}
\textnormal{See Fig.~6 for the graph of the concave envelope of $f^{\psi}_{\mathfrak{P}}$. Note that $f^{\psi}_{\mathfrak{P}}(\alpha)$ behaves as described in \S\ref{Props}, as expected.}
\end{eg}

\begin{eg}[Recovery of Classical Results]\label{eg:recovery}
\textnormal{Some of the results obtained by R.~Cawley and R.~D.~Mauldin in \cite{CM}, partially reproduced in \S\ref{Props}, along with the classical results obtained by A.~S.~Besicovitch in \cite{Bes1} and H.~G.~Eggleston in \cite{Egg} and described in Remark \ref{rk:besegg}, can be recovered in the setting of this section. For instance, see Example 1.6 on page 205 of \cite{CM}. Setting $N=4, \bfr=(r,r,r,r),$ and $\bfp=(p_1',p_1',p_2',p_2')$ (hence $w=2$) yields the desired recovery of this example. Another classical example which can be recovered in our setting is the binomial measure $\beta_0$ supported on the unit interval, as noted in Corollary \ref{cor:recoverbinom}. Indeed, setting $N=2, \bfr=(1/2,1/2)$ and $\bfp=(1/3,2/3)$ (hence $w=2$) yields the desired recovery in this instance.}
\end{eg}

\begin{eg}[Monofractal Measures]\label{eg:mono}
\textnormal{If there exists $D>0$ such that $r_i^D = p_i$ for all $i \in \{1,\ldots,N\}$, then the resulting measure $\rho$ is the natural Hausdorff measure of the underlying generalized Cantor set as discussed on page 201 of \cite{CM} and the end of \S\ref{Props}. More specifically, this Cantor set, which coincides with the support of $\rho$, is the self-similar set defined by the IFS $\{S_i\}_{i=1}^N$. The primitive of $\rho$ is defined by
\(\int_0^x(d\rho) = \rho([0,x])\) and its graph, at least in the usual case of the ternary Cantor set when $N=2$, $\bfr=(1/3,1/3)$, and $\bfp=(1/2,1/2)$, is the well-known Devil's staircase (see, for instance, \cite[Ch. 6]{Fed}). (In other words, the primitive of $\rho$ is the Cantor--Lebesgue function.) Recall that this function is nondecreasing and continuous on $[0,1]$ with zero derivative almost everywhere (on the Cantor string, in fact), yet its range is the full interval $[0,1]$.
See \cite[\S 12.2]{LapvF4} for an investigation of the Devil's staircase and a discussion of a new notion of fractality based on the distribution of complex dimensions of fractal strings. Fittingly, the only regularity values attained by $\rho$ on the intervals from its natural sequence of partitions are $\infty$ and $D$. Also, $f_{\mathfrak{P}}^{\rho}(D) = D$ and according to Definition \ref{dfn:abscon}, this abscissa of convergence function is trivially equal to zero for all other finite regularity values since no length stemming from $\mathfrak{P}$ has finite regularity value different from $D$. (Note that in the special case where $\mu=\rho$, as above, we have $D=\log_{3}2$.)}
\end{eg}

\ndnt The following example is the result of work done by Scott Roby in the Multifractal Analysis research project at California State University, Stanislaus in December 2010 and January 2011. The weighted IFS associated with this example satisfies the OSC, but the resulting regularity values are not distinct in the sense of hypothesis (\textbf{H}) from Theorem \ref{thm:distinctreg}, nor are the conditions of Lemma \ref{lem:regssm} satisfied. Nonetheless, the attained regularity values and the family of partition zeta functions are fully determined. As a special case, the geometric zeta function of the Fibonacci string $\zeta_{\textnormal{Fib}}(s)$ is recovered. (See \S\ref{CS} above for the development of the Fibonacci string and for a more thorough analysis, see \S 2.3.2 of \cite{LapvF4}.)

\begin{eg}[Recovery of the Fibonacci String]\label{eg:scottroby}
\textnormal{Consider a weighted IFS that satisfies the OSC where $\bfr=(1/2,1/4,1/10)$ and $\bfp=(1/2,1/4,1/4)$. Then the regularity values attained by the resulting measure $\mu$ with respect to its natural sequence of partitions $\mathfrak{P}$ are given by} 
\begin{equation*}\label{eq:alphascott}
\alpha{(\textnormal{\bfk})}=\alpha(k_1,k_2,k_3) = \frac{\log{2^{-(k_1+2k_2+2k_3)}}}{\log{(2^{-(k_1+2k_2)}10^{-k_3})}},
\end{equation*}
\textnormal{where $k_1,k_2,k_3 \in \N$ and are not all zero. Note that the vectors $\textnormal{\bfk}_1 = (1,0,0)$ and $\textnormal{\bfk}_2 = (0,1,0)$ both yield regularity value 1. Hence, the collection of attained regularity values are not distinct in the sense of hypothesis (\textbf{H}) from Theorem \ref{thm:distinctreg}. In order to distinguish the regularity values in this case, set  
$M:=k_1+2k_2$ and suppose $k_3 \neq 0$. Define $\alpha(M,k_3)$ by}

\begin{equation*}\label{eq:alphascottnew}
\alpha(M,k_3):=\alpha(\textnormal{\bfk}) = 1 + \frac{\log{2}-\log{5}}{\frac{M}{k_3}\log{2}+\log{10}}.
\end{equation*} 

\textnormal{The values of $\alpha(M,k_3)$ are distinct when $\gcd(M,k_3)=1$, which we assume to be the case for the remainder of this example. For a given regularity value $\alpha(M,k_3)$ and positive integer $n$, the corresponding $\alpha$-lengths are}
\begin{equation*}\label{eq:lscott}
l_n(\alpha(M,k_3)):= (2^{-M}10^{-k_3})^n,
\end{equation*}
\textnormal{with multiplicities given by}
\begin{equation*}\label{eq:mscott}
m_n(\alpha(M,k_3)):= \sum_{i=0}^{\left\lfloor\frac{nM}{2}\right\rfloor} \binom{\left\lfloor\frac{nM+2nk_3+1}{2}+i\right\rfloor}{nM-2\left\lfloor\frac{nM}{2}\right\rfloor+2i, \left\lfloor\frac{nM}{2}\right\rfloor-i,nk_3},
\end{equation*}
\textnormal{where $\left\lfloor \cdot \right\rfloor$ is the floor (i.e., integer part) function. That is, for $x \in \R$, $\left\lfloor x \right\rfloor$ is the greatest integer such that $\left\lfloor x \right\rfloor \leq x$. The partition zeta functions are then given by}
\begin{equation*}\label{eq:pzfscott}
\zeta^{\mu}_{\mathfrak{P}}(\alpha(M,k_3),s)=\sum_{n=1}^{\infty}m_n(\alpha(M,k_3))(l_n(\alpha(M,k_3)))^s.
\end{equation*}
\textnormal{In the case where $M=1$ and $k_3=0$, so that $\alpha(1,0)=1$, we recover the geometric zeta function of the Fibonacci string up to an additive constant (or multiplication by a nowhere vanishing entire function in terms of the closed forms of these functions). Indeed, in this case the lengths are given by $l_n(1)=2^{-n}$ and the multiplicities are given by}
\begin{equation*}\label{eq:recoverfibnumbers}
\ndnt m_n(1)= \sum_{i=0}^{\left\lfloor\frac{n}{2}\right\rfloor} \binom{\left\lfloor\frac{n+1}{2}+i\right\rfloor}{n-2\left\lfloor\frac{n}{2}\right\rfloor+2i}=F_{n+1},
\end{equation*}
\textnormal{where $F_{n+1}$ is the $(n+1)$th Fibonacci number. We refer the reader to \S\ref{CS} above and to \cite[\S2.3.2]{LapvF4} for a discussion of the Fibonacci string and its geometric zeta function $\zeta_{\textnormal{Fib}}$. The partition zeta function is therefore given by}
\begin{equation*}\label{eq:recovergzffib}
\zeta^{\mu}_{\mathfrak{P}}(1,s)=\sum_{n=1}^{\infty}F_{n+1}2^{-ns}=\zeta_{\textnormal{Fib}}(s)-1=\frac{2^{-s}+4^{-s}}{1-2^{-s}-4^{-s}},
\end{equation*}
\textnormal{where $\zeta_{\textnormal{Fib}}$ is given by Eq.~\eqref{eq:gzffib}. Thus, the corresponding complex dimensions, in both the classic sense and with respect to the parameter $1$, are given by (cf. Eq.~\eqref{eq:cdfib})}
\begin{equation*}\label{eq:recovercdfib}
\mathcal{D}_{\textnormal{Fib}}=\mathcal{D}^{\mu}_{\mathfrak{P}}(1)
=\left\{D+jzp \mid z \in \Z \right\} \cup \left\{-D+j(z+1/2)p \mid z \in \Z \right\},
\end{equation*}
\textnormal{where $\phi=(1+\sqrt{5})/2$ is the Golden Ratio, $D=\log_2{\phi}$, and $p = 2\pi/\log{2}$.}

\ndnt \textnormal{The only other case in this example for which the partition zeta function and the complex dimensions are explicitly known is when $M=0$ and $k_3=1$. In this case, the partition zeta function is}
\begin{equation*}\label{eq:pzfscotteasy}
\zeta^{\mu}_{\mathfrak{P}}(\alpha(0,1),s)= \sum_{n=1}^{\infty}10^{-ns}.
\end{equation*}
\textnormal{The complex dimensions with respect to the regularity value $\alpha(0,1)$ are therefore given by}
\begin{equation*}\label{eq:cdscotteasy}
\mathcal{D}^{\mu}_{\mathfrak{P}}(\alpha(0,1))=\left\{jzp \mid z \in \Z \right\}, 
\end{equation*}
\textnormal{where $p=2\pi/\log{10}$.}

\ndnt \textnormal{The remaining complex dimensions with respect to an arbitrary regularity value $\alpha$, and even the corresponding abscissae of convergence, have yet to be determined.}
\end{eg}

\begin{rk}\label{rk:developtheory}
\textnormal{The next step in the development of the theory of complex dimensions for self-similar measures is to determine $\mathcal{D}^{\mu}_{\mathfrak{P}}(\alpha, W_{\alpha})$ and $\mathcal{T}_{\mathfrak{P}}^{\mu}$, in the general situation considered in \S\ref{pzfssm} (or at least for interesting classes of examples, such as the multinomial multifractal measures). The work of D. Essouabri and the second author in \cite{EssLap} should provide a solid foundation for such a pursuit. It suggests, in particular, that the theory of complex fractal dimensions developed in \cite{LapvF1} and \cite{LapvF4} (or \cite{LapvF6}) should eventually be extended to apply to zeta functions that are viewed as analytic functions on Riemann surfaces (rather than just on suitable domains of the complex plane $\C$ or of the Riemann sphere $\C^*=\C\cup\{\infty\}$). In the present situation, the classic Riemann surface associated with the logarithm (or the square root) would be required; see \cite{EssLap}, which is motivated in part by the earlier, less general, results obtained in \cite{LapRo1,Rock} and described in
\cite[\S 13.3.6]{LapvF6}.}
\end{rk}

\section{Partition Zeta Functions of Atomic Measures}\label{pzfam}

In this section, we investigate the properties of certain atomic measures which are not self-similar in the sense of \S\ref{pzfssm}. Let $\sigma_1$ be given by
\[
\sigma_1 = \sum_{i=1}^{\infty} 3^{-i} \delta_{3^{-i}}
\]
and let $\Omega_1 = (0,1) \backslash \{3^{-i}\}_{i=1}^{\infty}$ be the fractal string determined by the complement in $[0,1]$ of the support of $\sigma_1$ (less the point 1). Let $\Omega_2$ be the open subset of $[0,1]$ obtained by placing disjoint open intervals with the lengths of the Cantor string $\mathcal{L}_{CS}$ end-to-end in nonincreasing order from right to left, with the single interval of length $1/3$ placed so that its right-endpoint is at 1 (see Fig.~7). Then, let $\sigma_2$ be the atomic measure supported on the left-endpoints of the fractal string $\Omega_2$, where the left-endpoint of each distinct open interval has weight given by the length $l_n = 3^{-n}$ of said interval (see Fig.~7).

\ndnt The sequence of distinct lengths $\mathcal{L}_1$ of the fractal string $\Omega_1$ is given by $\mathcal{L}_1 = \{2\cdot3^{-n}\}_{n=1}^{\infty}$, where each length has multiplicity 1. Also, the sequence of distinct lengths $\mathcal{L}_2$ of the fractal string $\Omega_2$ is exactly the same as the sequence of distinct lengths of the Cantor string. More specifically, $\mathcal{L}_2=\mathcal{L}_{CS}=\{3^{-n}\}_{n=1}^{\infty}$, but where each length $3^{-n}$ has multiplicity $2^{n-1}$ (instead of $1$). See Fig.~7 and \S\ref{CS}.

\begin{figure}
    \epsfysize=2.8cm\epsfbox{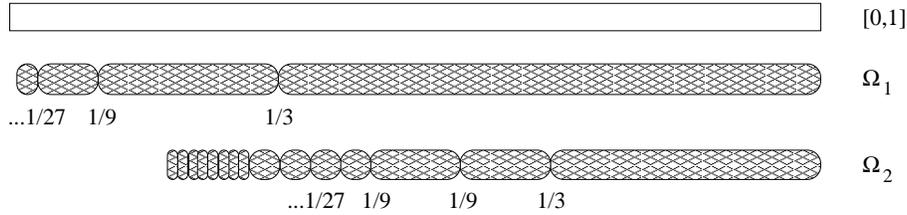}\label{atomicmeasures}
    \caption{\textit{Approximation of the fractal strings $\Omega_1$ and $\Omega_2$.}}
\end{figure}

\ndnt In order to determine the corresponding partition zeta functions for $\sigma_1$ and $\sigma_2$, we must choose a suitable sequence of partitions. In the absence of naturally defined sequences of partitions for $\sigma_1$ and $\sigma_2$, throughout this section we take $\mathfrak{P}$ to be the sequence of partitions $\mathcal{P}_n$ of left-closed, right-open ternary intervals $P_n^k$ of length $3^{-n}$ for $k \in \{1,\ldots,3^{n}-1\}$ and $P_n^{3^n}=[(3^{n}-1)/3^n,1]$, ordered from left to right by $k$. That is, for each $n \in \N^*$, the partition $\mathcal{P}_n$ is given by
\begin{equation}\label{eq:partitions}
\{[0,1/3^{n}),[1/3^{n},2/3^{n}),\ldots,[(3^{n}-2)/3^n,(3^{n}-1)/3^n),[(3^{n}-1)/3^n,1]\}.
\end{equation}

\subsection{A Full Family of Multifractal Complex Dimensions}\label{FFMCD}

As with the determination of the other partition zeta functions in this paper, the most delicate part of the process in the case of the measures $\sigma_1$ and $\sigma_2$ and the sequence of partitions $\mathfrak{P}$ is to find and distinguish the nontrivial regularity values.

\begin{lem}\label{lem:sigma1reg}
For the measure $\sigma_1$ and the sequence of partitions $\mathfrak{P}$ as given in Eq.~\eqref{eq:partitions}, the distinct nontrivial regularity values have the following forms: $1+\log_{3^n}2, k_1/K,$ and $\infty$, where $n, k_1, K \in \N^*, k_1 \leq K,$ and $\gcd(k_1,K)=1$.
\end{lem}

\begin{proof}
For each $n\in\N^*$, the leftmost interval $P^1_n$ of each partition $\mathcal{P}_n$ has regularity given by $A(P^1_n)=1+\log_{3^n}2$. For $n, k_1, K \in \N^*, k_1 \leq K,$ and $\gcd(k_1,K)=1$, the intervals $P^{3^{k_1n}}_{Kn}$ have regularity given by $A(P^{3^{k_1n}}_{Kn})=k_1/K$. No other interval $P^k_n$ stemming from $\mathfrak{P}$ has mass, thus $A(P^k_n)=\infty$ for each of these intervals. See Fig.~8.
\end{proof}

\begin{figure}
\epsfysize=5.1cm\epsfbox{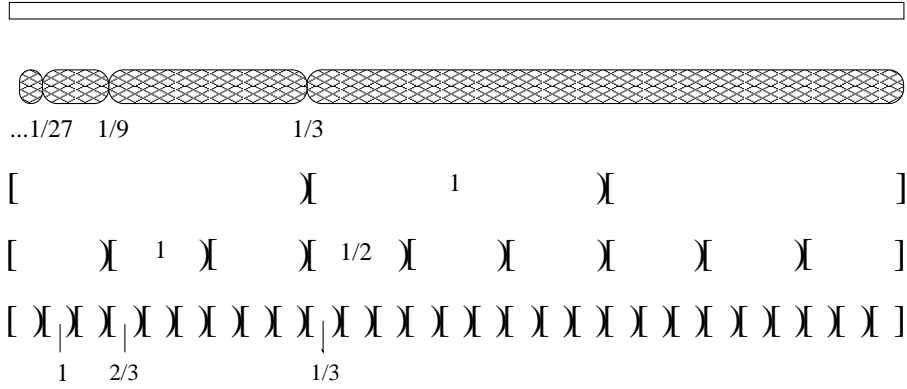}
    \caption{\textit{A breakdown of the regularity values attained by $\sigma_1$ with respect to the sequence of partitions $\mathfrak{P}$. The value inside each interval is its regularity. The leftmost blank intervals have regularity $\alpha=1+\log_{3^n}2$ at each stage $n \in \N^*$. The remaining blank intervals have no mass and therefore have regularity $\alpha=\infty$.}}
\end{figure}

\begin{rk}\label{rk:sigma1lengths}
\textnormal{An immediate consequence of Lemma \ref{lem:sigma1reg} is that the $\alpha$-lengths of $\sigma_1$ for $\alpha = k_1/K$ are given by}
\begin{equation}\label{eq:lsigma1}
\mathcal{L}_{\mathfrak{P}}^{\sigma_1}(k_1/K) = \{ 3^{-Kn} \mid 3^{-Kn} \textnormal{ has multiplicity } 1, n \in \N^*\};
\end{equation}
\textnormal{see Fig.~8. The sequences $\mathcal{L}_{\mathfrak{P}}^{\sigma_1}(k_1/K)$ are strongly languid and trivially self-similar\footnote{Strictly speaking, they are not self-similar since a single scaling ratio is involved and so the multiplicities are trivial.} (see Remark \ref{rk:stronglylanguid}). This fact effects the forms of the explicit formulas for the counting functions of the $\alpha$-lengths to be presented in Theorem \ref{thm:countingsigma1}.}
\end{rk}

\ndnt Before stating and proving Theorem \ref{thm:countingsigma1}, we give the forms of the partition zeta functions, abscissa of convergence function, complex dimensions with respect to $\alpha$, and tapestry of complex dimensions corresponding to the nontrivial and finite regularity values $\alpha$ obtained by $\sigma_1$ with respect to $\mathfrak{P}$. At this point, the reader may wish to briefly review \S\ref{Definitions}, specifically the definition of the complex dimensions with parameter $\alpha$ (Definition \ref{dfn:poles}) and the definition of the tapestry of complex dimensions (Definition \ref{dfn:tapestry}).

\begin{thm}\label{thm:sigma1}
The partition zeta functions for the nontrivial and finite regularity values $\alpha$ obtained by $\sigma_1$ with respect to $\mathfrak{P}$ as in Lemma \ref{lem:sigma1reg} are respectively given by
\begin{equation}\label{eq:pzfsigma11}
\zeta^{\sigma_1}_{\mathfrak{P}}(1+\log_{3^n}2,s) = 3^{-ns}
\end{equation}
and
\begin{equation}\label{eq:pzfsigma1k}
\zeta^{\sigma_1}_{\mathfrak{P}}(k_1/K,s) = \frac{3^{-Ks}}{1-3^{-Ks}},
\end{equation}
where  $s \in \C, n, k_1, K \in \N^*, k_1 \leq K,$ and $\gcd(k_1,K)=1$.

\ndnt Furthermore, the abscissa of convergence function is given by
\begin{equation}\label{eq:absconsigma1}
f^{\sigma_1}_{\mathfrak{P}}(\alpha) = 0, \hs \textnormal{for all} \hs \alpha \in (-\infty,\infty).
\end{equation}

\ndnt Moreover, for regularity $\alpha = k_1/K$, the complex dimensions with respect to $\alpha$ are given by
\begin{equation}\label{eq:cdssigma1}
\mathcal{D}_{\mathfrak{P}}^{\sigma_1}(k_1/K,\C) = \left\{ \frac{2 \pi j z}{K\log{3}} \mid z \in \Z \right\}.
\end{equation}

\ndnt Lastly, the tapestry of complex dimensions $\mathcal{T}_{\mathfrak{P}}^{\sigma_1}$ is given by
\begin{equation}\label{eq:tapsigma1}
\mathcal{T}_{\mathfrak{P}}^{\sigma_1} = \left\{ (\alpha,\omega) \mid \alpha = \frac{k_1}{K}, \omega \in 
\mathcal{D}_{\mathfrak{P}}^{\sigma_1}(k_1/K,\C), \textnormal{ where } k_1 \leq K, k_1,K \in \mathbb{N}^* \right\},
\end{equation}
as portrayed in Fig.~9.
\end{thm}

\begin{figure}
\epsfysize=7cm\epsfbox{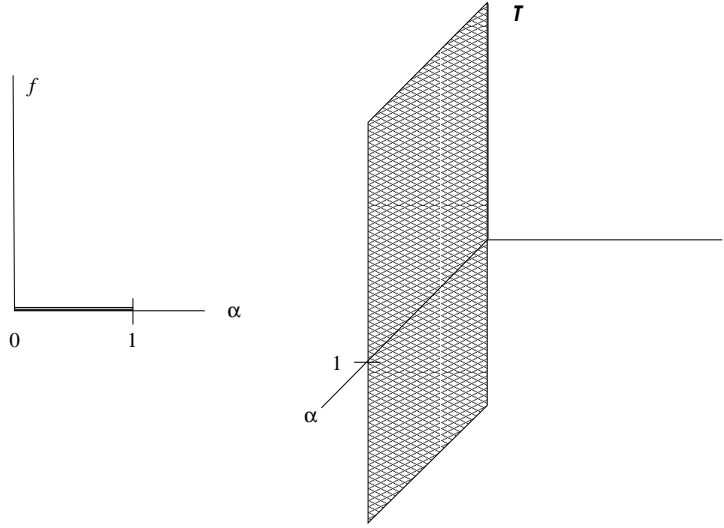}\label{sigma1tapestry}
    \caption{\textit{The abscissa of convergence function $f = f^{\sigma_1}_{\mathfrak{P}}(\alpha)=0$ on $(0,1]$ \textnormal{(}left\textnormal{)} and tapestry of complex dimensions $\mathcal{T} = \mathcal{T}_{\mathfrak{P}}^{\sigma_1}$ of $\sigma_1$ with respect to $\mathfrak{P}$ \textnormal{(}right\textnormal{)}.}}
\end{figure}

\begin{proof} Eqs.~\eqref{eq:pzfsigma11} and \eqref{eq:pzfsigma1k}
follow immediately from Lemma \ref{lem:sigma1reg} and Remark \ref{rk:sigma1lengths}. Note that Eq.~\eqref{eq:pzfsigma1k}, explicitly
\[
\zeta^{\sigma_1}_{\mathfrak{P}}(k_1/K,s) = \sum_{n=1}^{\infty} 3^{-Kns} = \frac{3^{-Ks}}{1-3^{-Ks}},
\]
is first obtained by assuming that $\textnormal{Re}(s)>0$, so that the geometric series involved converges. However, the end result clearly holds for all $s \in \C$, as can be seen upon meromorphic continuation. Hence, Eq.~\eqref{eq:pzfsigma1k} holds for all $s \in \C$.

\ndnt Observe that the partition zeta functions with regularity $\alpha=1+\log_{3^n}2$ for all $n \in \N^*$ have no poles and hence no complex dimensions, so their abscissae of convergence are trivially equal to $-\infty$; and hence, according to Definition \ref{dfn:abscon}, we have $f_{\mathfrak{P}}^{\sigma_1}(\alpha):=\max\{0,-\infty\}=0$ for all these regularity values $\alpha$. For those partition zeta functions with regularity $k_1/K$, the abscissa of convergence is also zero since $s=0$ is the unique real-valued solution to the equation $1-3^{-Ks}=0$. Therefore, $f^{\sigma_1}_{\mathfrak{P}}(\alpha) = 0$ for all $\alpha \in (-\infty,\infty)$; see Fig.~9.

\ndnt The expressions for the partitions zeta functions with regularity $\alpha=k_1/K$ have numerators which never vanish. Therefore, we deduce that the complex dimensions with parameter $\alpha$ are given by Eq.~\eqref{eq:cdssigma1}.
In turn, in light of Definition \ref{dfn:tapestry}, the complex dimensions with respect to $\alpha$ immediately yield the tapestry of complex dimensions given in Eq.~\eqref{eq:tapsigma1}.
\end{proof}

\begin{rk}\label{rk:reginfty} \textnormal{The only other regularity value attained by $\sigma_1$ with respect to $\mathfrak{P}$ is $\alpha=\infty$; see Fig.~8 along with \S\ref{CHR}. However, the intervals with such regularity are so numerous that the resulting partition zeta function $\zeta^{\sigma_1}_{\mathfrak{P}}(\infty,s)$ is divergent everywhere. A similar remark holds for the measure $\sigma_2$ discussed below in \S\ref{NLMS}.}
\end{rk}

\begin{rk}\label{rk:cdtap} \textnormal{The complex dimensions $\mathcal{D}_{\mathfrak{P}}^{\sigma_1}(k_1/K,\C)$ and the tapestry $\mathcal{T}_{\mathfrak{P}}^{\sigma_1}$ are exactly the same as those obtained for the measure $\sigma=\sigma_1$ in Corollary 13.55 and Remark 13.56 of \cite[\S 13.3.5]{LapvF6} (which was written in conjunction with the fourth author of this paper and describes joint work of the second and fourth authors). However, those results are obtained in the context provided by {\it multifractal zeta functions}, which we discuss briefly in \S\ref{conclusion}. The multifractal zeta functions are examined in \cite{LLVR,LapRo1,LapvF6,Rock} and are defined therein by a measure, a regularity value, and a sequence of scales (instead of a sequence of partitions). The multifractal structure of atomic measures similar to $\sigma_1$ and $\sigma_2$ are considered in \cite{LVT}, but not in the context of partition or multifractal zeta functions.}
\end{rk}

\ndnt Next, we give an explicit formula, expressed in terms of the underlying complex dimensions $\mathcal{D}_{\mathfrak{P}}^{\sigma_1}(\alpha,\C)$, for the counting functions of the $\alpha$-lengths of $\sigma_1$ with respect to $\mathfrak{P}$, as given by $\mathcal{L}_{\mathfrak{P}}^{\sigma_1}(\alpha)$ in Remark \ref{rk:sigma1lengths} above. 

\begin{thm}[Exact pointwise formula for the $\alpha$-lengths of $\sigma_1$]\label{thm:countingsigma1}
For each regularity value $\alpha=k_1/K$, with $k_1,K \in \N^*, k_1 \leq K,$ and $\gcd(k_1,K)=1$ \textnormal{(}as in Lemma \ref{lem:sigma1reg}\textnormal{)}, the counting function of the $\alpha$-lengths of $\sigma_1$ with respect to $\mathfrak{P}$ satisfies
\begin{equation}\label{eq:csigma1}
N_{\mathfrak{P}}^{\sigma_1}(\alpha, x) = N_{\mathfrak{P}}^{\sigma_1}(k_1/K, x) = M
= \frac{1}{K\log{3}} \sum_{\omega \in \mathcal{D}_{\alpha}}\frac{x^{\omega}}{\omega},
\end{equation}
where the formula holds pointwise for every $x>1$, with $M := \left\lfloor\log_{3^K}x\right\rfloor$ and  $\mathcal{D}_{\alpha}:=\mathcal{D}_{\mathfrak{P}}^{\sigma_1}(k_1/K,\C)$ as given in Eq.~\eqref{eq:cdssigma1} of Theorem \ref{thm:sigma1}. (Here, as before, $\left\lfloor y \right\rfloor$ denotes the integer part of $y$.)
\end{thm}

\begin{proof}
In this proof, we must assume that the reader has some familiarity with the theory developed in \cite{LapvF4}. Fix a regularity value $\alpha=k_1/K$ as given by Lemma \ref{lem:sigma1reg}. In light of Eq.~\eqref{eq:lsigma1}, the fact that \(N_{\mathfrak{P}}^{\sigma_1}(k_1/K,x) = M = \left\lfloor\log_{3^K}x\right\rfloor\) is immediate.

\ndnt Next, we justify the explicit formula \eqref{eq:csigma1} for $N_{\mathfrak{P}}^{\sigma_1}(k_1/K,x)$. This result follows from \cite[Thm.~5.14]{LapvF4}, the pointwise explicit formula without error term, applied at level $k=1$ (in the terminology of \cite{LapvF4}) to the zeta function
\[
\zeta^{\sigma_1}_{\mathfrak{P}}(\alpha,s)=\zeta^{\sigma_1}_{\mathfrak{P}}(k_1/K,s)=\frac{3^{-Ks}}{1-3^{-Ks}}, \hs
s \in \C,
\]
viewed as the ``geometric'' (or rather ``scaling'') zeta function of the generalized fractal string associated with the $\alpha$-lengths of $\sigma_1$ (see Definition \ref{dfn:alphalengths}). More specifically, since 0 is a pole of $
\zeta^{\sigma_1}_{\mathfrak{P}}(\alpha,\cdot)$ and with our present notation, \cite[Thm.~5.14]{LapvF4} yields for all $x>A$ (with $A:=1$, as explained below):
\begin{align*}
N_{\mathfrak{P}}^{\sigma_1}(\alpha, x) &= \sum_{\omega \in \mathcal{D}_{\alpha}} \textnormal{res}\left(\frac{x^s}{s}\zeta^{\sigma_1}_{\mathfrak{P}}(\alpha,s);s=\omega \right)  \\
&= \sum_{\omega \in \mathcal{D}_{\alpha}} \frac{x^{\omega}}{\omega}\textnormal{res}\left(\zeta^{\sigma_1}_{\mathfrak{P}}(\alpha,s);s=\omega \right).
\end{align*}
Hence \eqref{eq:csigma1} follows since $\zeta^{\sigma_1}_{\mathfrak{P}}(\alpha,s)$ and $\mathcal{D}_{\alpha}:=\mathcal{D}_{\mathfrak{P}}^{\sigma_1}(\alpha,\C)$
 are given by Eqs.~\eqref{eq:pzfsigma1k} and \eqref{eq:cdssigma1}, respectively, and consequently, for $\omega=2\pi j z/(K\log3)$ (with $z\in\Z$), we have
\[
\textnormal{res}\left(\zeta^{\sigma_1}_{\mathfrak{P}}(\alpha,s);s=\omega\right)= \textnormal{res}\left(\frac{3^{-Ks}}{1-3^{-Ks}};s=\omega\right)=
\frac{3^{-K\omega}}{(K\log3)3^{-K\omega}}=\frac{1}{K\log3},
\]
independently of $z \in \Z$.

\ndnt Note that the aforementioned explicit formula of \cite{LapvF4} can be applied here because an elementary computation (entirely analogous to the one performed on page 189 of \cite[\S 6.4]{LapvF4}) shows that $\zeta^{\sigma_1}_{\mathfrak{P}}(\alpha,\cdot)$ is strongly languid (in the sense of \cite[Def.~5.3]{LapvF4}) of order $\kappa = 0 <1$ and with constant $A=1$. (Here, we use the notation $\kappa$ and $A$ employed in \cite{LapvF4}; see especially \cite[\S 5.3]{LapvF4}.) More specifically, with $A:=1$, we clearly have
\[
\left|\frac{3^{-Ks}}{1-3^{-Ks}}\right|=\left|\frac{1}{1-3^{-Ks}}\right| \leq 1 = (A^{-1})^{-|\textnormal{Re}(s)|},
\]
as $\textnormal{Re}(s) \rightarrow -\infty$.\footnote{Strictly speaking, in the above inequality, 1 should be replaced by $\eta$, for any given $\eta>1$.} Also, we have $W=\C$ in this case. This concludes the proof of Theorem \ref{thm:countingsigma1}.
\end{proof}

\begin{rk}\label{rk:counting1}
\textnormal{We leave it as an exercise for the interested reader to verify that given the simple form of the sequence of $\alpha$-lengths obtained in Remark \ref{rk:sigma1lengths}, it is possible to recover Eq.~\eqref{eq:csigma1} by a direct computation (also involving a conditionally convergent Fourier series, but no longer using \cite[Thm.~5.14]{LapvF4}). (In more complicated situations, however, we would have to use the exact explicit formula in \cite[Thm~5.14]{LapvF4}, or its counterpart with error term given in \cite[Thm.~5.10]{LapvF4}, or even more generally, their distributional analogues obtained in \cite[\S 5.4]{LapvF4}.) We note that the computation would then resemble the one carried out in a related context for the Cantor string in \cite[\S 1.1.2]{LapvF4}. It is also useful to observe that Eq.~\eqref{eq:csigma1} can be equivalently rewritten as follows:}
\begin{equation}
N_{\mathfrak{P}}^{\sigma_1}(\alpha, x) = g(u), \hs \textnormal{with } u:=\left\lfloor\log_{3^K}x\right\rfloor,
\end{equation}

\textnormal{where $g$ is the 1-periodic function given by the (conditionally) convergent Fourier series}
\begin{equation}
g(u):=\frac{1}{2\pi j}\sum_{z \in \Z}\frac{e^{2\pi jzu}}{z}, \hs u \in \R.
\end{equation}
\end{rk}

\ndnt Observe that the lack of positive real part in the complex dimensions $\mathcal{D}_{\mathfrak{P}}^{\sigma_1}(k_1/K,\C)$ stems from the unit multiplicity of each corresponding distinct $\alpha$-length in $\mathcal{L}_{\mathfrak{P}}^{\sigma_1}(k_1/K)$. In the case of $\sigma_2$, however, the multiplicities of the distinct $\alpha$-lengths are integer powers of 2. This results in a nonconstant linear multifractal spectrum for $\sigma_2$, as described in the next section.

\subsection{A Nonconstant Linear Multifractal Spectrum}\label{NLMS}

\begin{figure}
\epsfysize=5.1cm\epsfbox{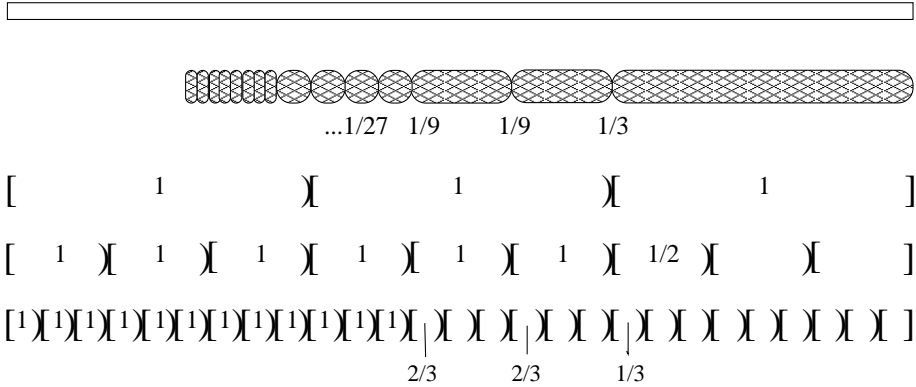}\label{sigma2reg}
    \caption{\textit{A breakdown of the regularity values attained by $\sigma_2$ with respect to the sequence of partitions $\mathfrak{P}$. The blank intervals have no mass and therefore have regularity $\alpha=\infty$.}}
\end{figure}

The determination of the distinct nontrivial regularity values in the case of $\sigma_2$ with respect to $\mathfrak{P}$ is actually easier than that of $\sigma_1$.

\begin{lem}\label{lem:sigma2reg}
For the measure $\sigma_2$ and the sequence of partitions $\mathfrak{P}$ given at the beginning of \S\ref{pzfam}, the distinct nontrivial regularity values have the following forms: $\frac{\log(3^{-k_1})}{\log(3^{-K})} = \frac{k_1}{K},$ and $\infty$, where $k_1, K \in \N^*, k_1 \leq K,$ and $\gcd(k_1,K)=1$. \textnormal{(}See Fig.~10.\textnormal{)}
\end{lem}

\begin{proof}
Each interval $P^k_{Kn}$ which contains two or more point-masses has regularity 1, since $\sigma_2(P^k_{Kn}) = |P^k_{Kn}|$ for all such intervals. Furthermore, each interval of length $3^{-nK}$ which contains a single point-mass of the form $3^{-nk_1}$, with $n \in \N^*$, has regularity
\[
\alpha = \frac{\log(3^{-k_1})}{\log(3^{-K})} = \frac{k_1}{K}.
\]
Finally, observe that no other interval stemming from $\mathfrak{P}$ has mass and hence, $\alpha=\infty$ for each of these intervals; see Fig.~10.
\end{proof}

\begin{rk}\label{rk:sigma2lengths}
\textnormal{The sequence of $\alpha$-lengths $\mathcal{L}_{\mathfrak{P}}^{\sigma_2}(1)$ of $\sigma_2$ for $\alpha=k_1=K=1$ is given by}
\begin{equation}\label{eq:lengths1sigma2}
\mathcal{L}_{\mathfrak{P}}^{\sigma_2}(1) = \{ 3^{-n} \mid 3^{-n} \textnormal{ has multiplicity } 3\cdot2^{n-1}, n \in \N^*\}.
\end{equation}
\textnormal{Moreover, the sequence of $\alpha$-lengths $\mathcal{L}_{\mathfrak{P}}^{\sigma_2}(k_1/K)$ of $\sigma_2$ for $\alpha = k_1/K$ with $k_1<K$ is given by}
\begin{equation}\label{eq:lengthsksigma2}
\mathcal{L}_{\mathfrak{P}}^{\sigma_2}(k_1/K) = \{ 3^{-Kn} \mid 3^{-Kn} \textnormal{ has multiplicity } 2^{k_1n-1}, n \in \N^*\}.
\end{equation}
\textnormal{As with the case of $\sigma_1$ above, the sequences $\mathcal{L}_{\mathfrak{P}}^{\sigma_2}(k_1/K)$ are also self-similar and strongly languid; see Remark \ref{rk:stronglylanguid}. This effects the form of the counting function of the $\alpha$-lengths presented in Theorem \ref{thm:countingsigma2} below.}
\end{rk}

\ndnt Before stating and proving Theorem \ref{thm:countingsigma2}, we give the forms of the partition zeta functions, abscissa of convergence function, complex dimensions with respect to $\alpha$, and tapestry of complex dimensions corresponding to the nontrivial and finite regularity values $\alpha$ obtained by $\sigma_2$ with respect to $\mathfrak{P}$.

\begin{figure}
\epsfysize=7cm\epsfbox{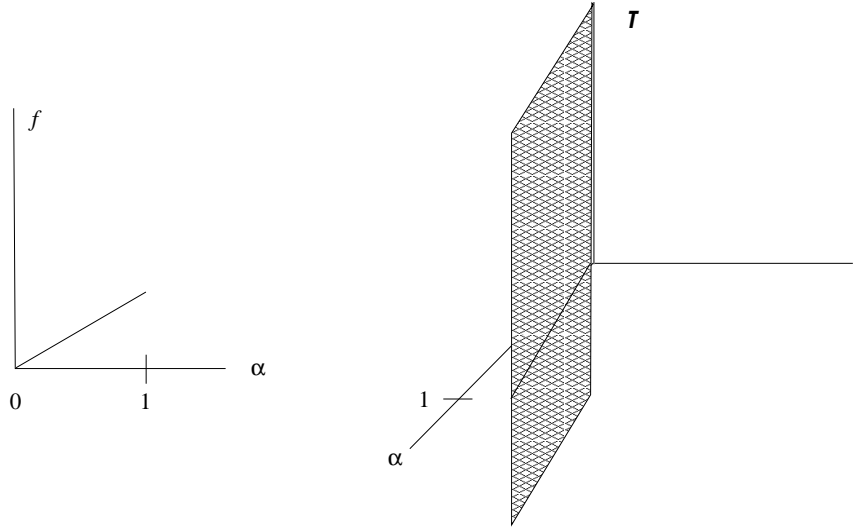}\label{sigma2spectap}
    \caption{\textit{The abscissa of convergence function $f = f^{\sigma_2}_{\mathfrak{P}}(\alpha)$ \textnormal{(}left\textnormal{)} and tapestry $\mathcal{T} = \mathcal{T}_{\mathfrak{P}}^{\sigma_2}$ of complex dimensions of $\sigma_2$ with respect to $\mathfrak{P}$ \textnormal{(}right\textnormal{)}.}}
\end{figure}

\begin{thm}\label{thm:sigma2}
The partition zeta functions for the nontrivial and finite regularity values $\alpha$ obtained by $\sigma_2$ with respect to $\mathfrak{P}$ as in Lemma \ref{lem:sigma2reg} are respectively given by
\begin{equation}\label{eq:pzf1sigma2}
\zeta^{\sigma_2}_{\mathfrak{P}}(1,s) = \frac{3 \cdot 3^{-s}}{1-2 \cdot 3^{-s}}
\end{equation} 
and
\begin{equation}\label{eq:pzfksigma2}
\zeta^{\sigma_2}_{\mathfrak{P}}(k_1/K,s) = \frac{2^{k_1-1}\cdot3^{-Ks}}{1-2^{k_1}\cdot 3^{-Ks}},
\end{equation}
where  $s \in \C, k_1, K \in \N^*, k_1 < K,$ and $\gcd(k_1,K)=1$. \textnormal{(}Here and henceforth, when $k_1<K$, we assume that $\gcd(k_1,K)=1$.\textnormal{)}

\ndnt Furthermore, the abscissa of convergence function is given by
\begin{equation}\label{eq:absconsigma2}
f^{\sigma_2}_{\mathfrak{P}}(\alpha) =
f^{\sigma_2}_{\mathfrak{P}}(k_1/K) = \frac{k_1}{K}\log_3{2},
\end{equation}
where $\alpha = k_1/K$ such that $k_1 \leq K$ with $k_1,K \in \mathbb{N}^*$; see Fig.~11. \textnormal{(}Note that this equation also holds when $\alpha=k_1=K=1$.\textnormal{)}

\ndnt Moreover, the set of complex dimensions $\mathcal{D}_{\mathfrak{P}}^{\sigma_2}(\alpha,\C)$ with respect to $\alpha=k_1/K$, where $k_1 \leq K$ with $k_1,K \in \mathbb{N}^*$, is
\begin{equation}\label{eq:cdssigma2}
\mathcal{D}_{\mathfrak{P}}^{\sigma_2}(k_1/K,\C)=\left\{ \frac{k_1}{K}\log_3{2} + \frac{2 \pi j z}{K\log{3}} \mid z \in \Z \right\}.
\end{equation}
\textnormal{(}Note that this equation still holds when $\alpha=k_1=K=1$.\textnormal{)}

\ndnt Lastly, the tapestry of complex dimensions $\mathcal{T}_{\mathfrak{P}}^{\sigma_2}$  is given by
\begin{equation}\label{eq:tapsigma2}
\mathcal{T}_{\mathfrak{P}}^{\sigma_2} = \left\{ (\alpha,\omega) \mid \alpha = \frac{k_1}{K}, \omega \in 
\mathcal{D}_{\mathfrak{P}}^{\sigma_2}(k_1/K,\C), \textnormal{ where } k_1 \leq K, k_1,K \in \mathbb{N}^* \right\},
\end{equation}
as portrayed in Fig.~11.
\end{thm}

\begin{proof}
This proof is very similar to that of Theorem \ref{thm:sigma1}. Indeed, in light of Lemma \ref{lem:sigma2reg}, Eqs.~\eqref{eq:pzf1sigma2} and \eqref{eq:pzfksigma2} follow from Eqs.~\eqref{eq:lengths1sigma2} and \eqref{eq:lengthsksigma2}, respectively, first by summing a geometric series (for $\textnormal{Re}(s)> \frac{k_1}{K}\log_3{2}$) and then by meromorphically continuing the resulting expressions to all of $\C$. We deduce at once that the abscissa of convergence function is given by Eq.~\eqref{eq:absconsigma2}.

\ndnt Finally, Eq.~\eqref{eq:cdssigma2} (and then Eq.~\eqref{eq:tapsigma2}) follows immediately from Eqs.~\eqref{eq:pzf1sigma2} and \eqref{eq:pzfksigma2}. 
\end{proof}

\begin{rk}\label{rk:sigma2envelope}
\textnormal{The concave envelope $\hat{f}^{\sigma_2}_{\mathfrak{P}}$ of $f^{\sigma_2}_{\mathfrak{P}}$ on $[0,1]$ is a nonconstant linear multifractal spectrum for the measure $\sigma_2$. Indeed, for every $t \in [0,1]$, we have}
\begin{equation}
\hat{f}^{\sigma_2}_{\mathfrak{P}}(t) = t \log_3{2}.
\end{equation}
\textnormal{Note that, unlike for the results of Theorem \ref{thm:distinctreg}, each nonzero value $\hat{f}^{\sigma_2}_{\mathfrak{P}}(t)$ cannot be equal to the Hausdorff dimension of some subset of the support of $\sigma_2$. Indeed, the support of $\sigma_2$ is the countable set $\partial\Omega_2$, which has Hausdorff dimension equal to zero.}
\end{rk}

\ndnt Next, we close this section by giving an explicit formula, expressed in terms of the underlying complex dimensions $\mathcal{D}_{\mathfrak{P}}^{\sigma_2}(\alpha,\C)$, for the counting functions of the $\alpha$-lengths of $\sigma_2$ with respect to $\mathfrak{P}$, as given by $\mathcal{L}_{\mathfrak{P}}^{\sigma_2}(\alpha)$ in Remark \ref{rk:sigma2lengths}.

\begin{thm}[Exact pointwise formula for the $\alpha$-lengths of $\sigma_2$]\label{thm:countingsigma2}
For each nontrivial regularity value of $\sigma_2$ with respect to $\mathfrak{P}$ given by Lemma \ref{lem:sigma2reg}, the counting function of the $\alpha$-lengths is as follows:
\begin{itemize}
\item[\textnormal{\textbf{1}}] \textnormal{($\alpha=k_1=K=1$).} For the regularity value $\alpha=1$, we have
\begin{equation}
N_{\mathfrak{P}}^{\sigma_2}(1, x) = 3\cdot(2^{M}-1)
= \frac{3}{2\log{3}} \sum_{\omega \in \mathcal{D}_{1}}\frac{x^{\omega}}{\omega}-3,
\end{equation}
where this formula holds pointwise for every $x>1$, with $M = \left\lfloor\log_{3}x\right\rfloor$ and $\mathcal{D}_{1}:=\mathcal{D}_{\mathfrak{P}}^{\sigma_2}(1,\C)$ as given in Eq.~\eqref{eq:cdssigma2}.

\item[\textnormal{\textbf{2}}] \textnormal{($\alpha=k_1/K$).} For the regularity value $\alpha=k_1/K$ such that $k_1<K$ with $k_1,K \in \N^*$ and $\gcd(k_1,K)=1$ \textnormal{(}as in Theorem \ref{thm:sigma2}\textnormal{)}, we have
\begin{align}
N_{\mathfrak{P}}^{\sigma_2}(\alpha, x) = N_{\mathfrak{P}}^{\sigma_2}(k_1/K, x) &= \frac{2^{k_1-1}(2^{k_1M}-1)}{2^{k_1}-1} \notag \\
&= \frac{1}{2\log3^K}\sum_{\omega \in \mathcal{D}_{\alpha}}
\frac{x^{\omega}}{\omega} + \frac{2^{k_1-1}}{1-2^{k_1}},
\end{align}
where this formula holds pointwise for every $x>1$, with $M := \left\lfloor\log_{3^K}x\right\rfloor$ \textnormal{(}as in Theorem \ref{thm:countingsigma1}\textnormal{)} and $\mathcal{D}_{\alpha}:=\mathcal{D}_{\mathfrak{P}}^{\sigma_2}(\alpha,\C)$ as given in Eq.~\eqref{eq:cdssigma2}.
\end{itemize}
\end{thm}

\begin{proof}
The proof parallels that of Theorem \ref{thm:countingsigma1} and therefore follows from \cite[Thm.~5.14]{LapvF4} by showing (as on page 189 of \cite[\S 6.4]{LapvF4}) that for each given regularity $\alpha$, $\zeta_{\mathfrak{P}}^{\sigma_2}(\alpha,\cdot)$ is strongly languid of order $\kappa=0$ and with constant $A=A_{\alpha}=1$. This last conclusion follows from the estimate
\[
\left|\zeta_{\mathfrak{P}}^{\sigma_2}(\alpha,s)\right| \ll (A^{-1})^{-|\textnormal{Re}(s)|},
\]
as $\textnormal{Re}(s) \rightarrow -\infty$, where $A=A_{\alpha}$ is given respectively by $A:=3^{-1}(3^{-1})^{-1}=1$
when $\alpha=1$ (as in Case {\bf 1}) and by $A:=3^{-K}(3^{-K})^{-1}=1$ with $\alpha=k_1/K$ (as in Case {\bf 2}).

\ndnt Moreover, since in either Case {\bf 1} or {\bf 2} of the theorem, 0 does not belong to $\mathcal{D}_{\alpha}$, the exact explicit formula of \cite[Thm.~5.14]{LapvF4} (applied at level $k=1$) yields for all $x>A_{\alpha}=1$:
\begin{align*}
N_{\mathfrak{P}}^{\sigma_2}(\alpha,x) &= \sum_{\omega \in \mathcal{D}_{\alpha}} \textnormal{res}\left(\frac{x^s}{s}\zeta_{\mathfrak{P}}^{\sigma_2}(\alpha,s); s=\omega \right) +
\zeta_{\mathfrak{P}}^{\sigma_2}(\alpha,0)\\
&= \sum_{\omega \in \mathcal{D}_{\alpha}} \frac{x^{\omega}}{\omega} \textnormal{res}\left(\zeta_{\mathfrak{P}}^{\sigma_2}(\alpha,s); s=\omega \right) +
\zeta_{\mathfrak{P}}^{\sigma_2}(\alpha,0).
\end{align*}
Consequently, the result follows since (in light of Theorem \ref{thm:sigma2}) an elementary computation shows that in Case {\bf 1} or {\bf 2}, respectively, we have for every $\omega \in \mathcal{D}_{\alpha}$:
\[
\zeta_{\mathfrak{P}}^{\sigma_2}(\alpha,0) =-3, \hs \textnormal{res}\left(\zeta_{\mathfrak{P}}^{\sigma_2}(\alpha,s); s=\omega \right) = \frac{3}{2\log3},
\]
while
\[
\zeta_{\mathfrak{P}}^{\sigma_2}(\alpha,0) = \frac{2^{k_1-1}}{1-2^{k_1}}, \hs \textnormal{res}\left(\zeta_{\mathfrak{P}}^{\sigma_2}(\alpha,s); s=\omega \right) = \frac{1}{2K\log3}.
\]
This concludes the proof of Theorem \ref{thm:countingsigma2}.
\end{proof}

\begin{rk}\label{rk:counting2}
\textnormal{A comment completely analogous to the one made in Remark \ref{rk:counting1} (for the measure $\sigma_1$) applies to the measure $\sigma_2$.}
\end{rk}

\ndnt In light of Theorem \ref{thm:sigma2}, we deduce at once the following result from Theorem \ref{thm:countingsigma2}:

\begin{cor}\label{cor:sigma2}
The expression for the counting functions for the $\alpha$-lengths can be rewritten as follows, in Case {\bf 1} $(\alpha=1)$ or Case {\bf 2} \textnormal{(}$\alpha=k_1/K$, with $k_1<K$ and $\gcd(k_1,K)=1$\textnormal{)} of Theorem \ref{thm:countingsigma2}:
\begin{align}
N_{\mathfrak{P}}^{\sigma_2}(1,x)+3 &= \frac{3x^{\log_3{2}}}{2\log3} \sum_{z\in\Z}\frac{x^{\frac{2\pi jz}{\log3}}}{\log_3{2}+ \frac{2\pi jz}{\log3}} \notag \\
&= x^{\log_3{2}} G_1(\log_3{x}) \notag \\
&= x^{f^{\sigma_2}_{\mathfrak{P}}(1)} G_1(\log_3{x}),
\end{align}
and \textnormal{(}for $\alpha=k_1/K$, as in Case {\bf 2}\textnormal{)}
\begin{align}
N_{\mathfrak{P}}^{\sigma_2}(\alpha,x)+\frac{2^{k_1-1}}{2^{k_1}-1}&=\frac{x^{k_1\log_{3^K}{2}}}{2\log3^K} \sum_{z\in\Z}\frac{x^{\frac{2\pi jz}{K\log3}}}{k_1\log_{3^K}{2}+ \frac{2\pi jz}{K\log3}} \notag \\
&= x^{k_1\log_{3^K}{2}} G_{\alpha}(\log_{3^K}{x}) \notag \\
&= x^{f^{\sigma_2}_{\mathfrak{P}}(\alpha)} G_{\alpha}(\log_{3^K}{x}),
\end{align}
where $G_1$ \textnormal{(}in Case {\bf 1}\textnormal{)} or $G_{\alpha}$ \textnormal{(}in Case {\bf 2}\textnormal{)} is the 1-periodic function given by the conditionally convergent Fourier series
\begin{equation}
G_1(u):= \frac{3}{2\log3}\sum_{z\in\Z}\frac{e^{2\pi jzu}}{\log_3{2}+ \frac{2\pi jz}{\log3}}, \hs u \in \R,
\end{equation}
and \textnormal{(}for $\alpha=k_1/K$, as in Case {\bf 2}\textnormal{)}
\begin{equation}
G_{\alpha}(u):= \frac{1}{2\log3^K}\sum_{z\in\Z}\frac{e^{2\pi jzu}}{\log_{3^K}{2}+ \frac{2\pi jz}{\log{3^K}}}, \hs u \in \R.
\end{equation}
\end{cor}

\ndnt Most of the key results of \S\ref{NLMS}, namely, Lemma \ref{lem:sigma2reg}, Theorem \ref{thm:sigma2} and Theorem \ref{thm:countingsigma2}, readily extend to cases involving slightly more general forms of the Cantor string. Such generalized results are stated here with less formality than in the previous sections and without proof since they follow those presented above, mutatis mutandis.

\begin{eg}[Generalization of the measure $\sigma_2$]\label{GenAtomic}
\textnormal{Let $m \in \N$ such that $m\geq 2$ and let $\lambda=(2m-1)^{-1}$. Using this $m$ and $\lambda$, construct the measure $\sigma$ in the way $\sigma_2$ is constructed. The multiplicity of the initial intervals is $m-1$ and each of these intervals has length $\lambda$. (We note that $m=2$ and hence $\lambda=1/3$ in the following development would allow us to recover the results from \S\ref{NLMS}.)}

\ndnt \textnormal{The sequence of partitions $\mathfrak{P}$ in this setting are partitions $\mathcal{P}_n$ which split the unit interval into disjoint subintervals of length $\lambda^n$. The nontrivial and finite regularity values $\alpha$ attained by $\sigma$ with respect to $\mathfrak{P}$ are exactly the same as those attained in the case of $\sigma_2$. That is,}
\begin{equation*}
\alpha(k_1,K) = \frac{\log(\lambda^{k_1})}{\log(\lambda^{K})} = \frac{k_1}{K},
\end{equation*}
\textnormal{and again we take $\alpha = k_1/K$ such that $k_1 \leq K$ with $k_1,K \in \mathbb{N}^*$ (and $\gcd(k_1,K) = 1$ when $k_1<K$). The sequence of $\alpha$-lengths $\mathcal{L}_{\mathfrak{P}}^{\sigma}(1)$ for $\alpha=k_1=K=1$ is given by}
\begin{equation*}
\mathcal{L}_{\mathfrak{P}}^{\sigma}(1) = \{ \lambda^{n} \mid \lambda^{n} \textnormal{ has multiplicity } \lambda^{-1}(m-1) m^{n-1}, n \in \N^*\}.
\end{equation*}
\textnormal{The sequence of $\alpha$-lengths $\mathcal{L}_{\mathfrak{P}}^{\sigma}(k_1/K)$ for $\alpha = k_1/K$ with $k_1<K$ is given by}
\begin{equation*}
\mathcal{L}_{\mathfrak{P}}^{\sigma}(k_1/K) = \{ \lambda^{Kn} \mid \lambda^{Kn} \textnormal{ has multiplicity } (m-1)m^{k_1n-1}, n \in \N^*\}.
\end{equation*}
\textnormal{Hence, the partition zeta functions are given by}
\begin{equation*}
\zeta^{\sigma}_{\mathfrak{P}}(1,s) = \frac{(m-1)\lambda^{-1}}{m}\sum_{n=1}^{\infty}(m \cdot \lambda^s)^n =
\frac{(m-1)\lambda^{s-1}}{1-m\lambda^s}
\end{equation*}
\textnormal{and}
\begin{equation*}
\zeta^{\sigma}_{\mathfrak{P}}(k_1/K,s) = \frac{(m-1)}{m} \sum_{n=1}^{\infty}(m^{k_1} \cdot \lambda^{Ks})^n =
\frac{(m-1)m^{k_1-1}\lambda^{Ks}}{1-m^{k_1}\lambda^{Ks}},
\end{equation*}
\textnormal{where  $s \in \C, k_1, K \in \N^*, k_1 < K,$ and $\gcd(k_1,K)=1$.}

\ndnt \textnormal{As a result, the abscissa of convergence function is given by}
\begin{equation*}
f^{\sigma}_{\mathfrak{P}}(\alpha) =
f^{\sigma}_{\mathfrak{P}}(k_1/K) = \frac{k_1}{K}\log_{\lambda^{-1}}{m},
\end{equation*}
\textnormal{where $\alpha = k_1/K$ such that $k_1 \leq K$ with $k_1,K \in \mathbb{N}^*$. Note that this equation also holds when $\alpha=k_1=K=1$. Furthermore, observe that $\log_{\lambda^{-1}}m>0$ since $\lambda^{-1}=2m-1$ and $m \geq 2$.}

\ndnt \textnormal{Moreover, the set of complex dimensions  $\mathcal{D}_{\mathfrak{P}}^{\sigma}(\alpha,\C)$ with respect to $\alpha=k_1/K$, where $k_1 \leq K$ with $k_1,K \in \mathbb{N}^*$, is}
\begin{equation*}
\mathcal{D}_{\mathfrak{P}}^{\sigma}(k_1/K,\C)=\left\{ \frac{k_1}{K}\log_{\lambda^{-1}}{m} + \frac{2 \pi j z}{K\log{\lambda^{-1}}} \mid z \in \Z \right\}.
\end{equation*}
\textnormal{Note that this equation also holds when $\alpha=k_1=K=1$.}

\ndnt \textnormal{It follows that the tapestry of complex dimensions $\mathcal{T}_{\mathfrak{P}}^{\sigma}$  is given by}
\begin{equation*}
\mathcal{T}_{\mathfrak{P}}^{\sigma} = \left\{ (\alpha,\omega) \mid \alpha = \frac{k_1}{K}, \omega \in 
\mathcal{D}_{\mathfrak{P}}^{\sigma}(k_1/K,\C), \textnormal{ where } k_1 \leq K, k_1,K \in \mathbb{N}^* \right\}.
\end{equation*}

\ndnt \textnormal{Finally, we leave it as an exercise for the interested reader to obtain the counterpart of Theorem \ref{thm:countingsigma2} (the explicit formula for the counting function of the $\alpha$-lengths), using the same line of reasoning as in the proof of that theorem.}
\end{eg}

\ndnt The following section concludes the paper with a brief description of natural questions and avenues of research provided by the approach to multifractal analysis via zeta functions adopted in this work.


\section{Conclusion}\label{conclusion}

The determination of the meromorphic continuation (or some other appropriate extension) of the partition
zeta functions will be addressed in the near future. Work in this direction has already begun by the second author and  D.~Essouabri in \cite{EssLap}. Upon a suitable change of variable, the results of such an investigation will provide the poles, and hence complex dimensions, for this family of self-similar multifractal measures. This leads naturally to the search for an understanding of multifractal objects in more general settings, specifically those with non-multiplicative construction and properties. In the long term, motivated by \cite{LLVR,LapRo1,Rock} and the theory of complex dimensions in \cite{LapvF1,LapvF4} (also with consideration of the work done by J.~L{\'e}vy~V{\'e}hel and F.~Mendivil in \cite{LVM}), one may wish to investigate the physical or geometric oscillations of multiplicative and non-multiplicative multifractal objects in geometric, spectral and dynamical settings, as was done with fractal strings by way of their complex dimensions. (See, for example, \cite{LapvF7,LapvF1,LapvF3,LapvF4,LapvF6}, along with the relevant references mentioned in the introduction.) This work, \cite{EssLap,LVM}, as well as the exposition of some of the aspects of \cite{LapRo1} given in \cite[\S 13.3]{LapvF6}, should provide a nice foundation for such a theory of complex dimensions for multifractals.

\ndnt A recent predecessor of this work is \cite{LLVR}, by M.~L.~Lapidus, J.~L{\'e}vy-V{\'e}hel and J.~A.~Rock, where certain Dirichlet series were introduced and used in order to study some geometric properties of fractal strings which are not accounted for in the
theory developed in \cite{LapvF1,LapvF4}. The intent of the
definition of the {\it multifractal zeta functions} from
\cite{LLVR,Rock}, however, was to extend the techniques used in the
theory of complex dimensions of fractal strings to multifractal
analysis in some way. An elaboration on the difficulties of using
these multifractal zeta functions to this end is provided in
\cite{Rock}, where the primary object study of this work, the
{\it partition zeta function}, is first introduced.

\begin{figure}
    \epsfysize=3.2cm\epsfbox{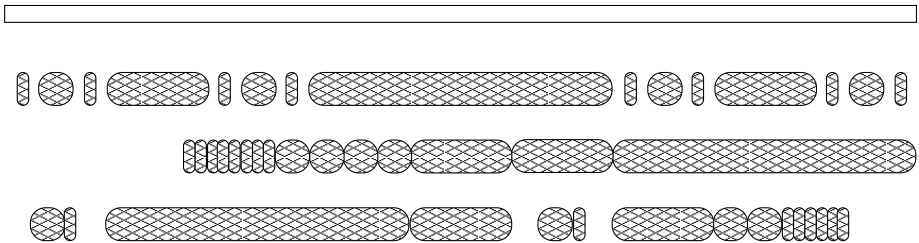}\label{threestrings}
    \caption{\textit{Approximation of three fractal strings with the same sequence of lengths $\mathcal{L}$, specifically the lengths of the Cantor string $\Omega_{CS}$. The boundary of each fractal string has the same Minkowski dimension, but the Hausdorff dimensions differ.}}
\end{figure}

\ndnt In \cite{Gibson,LLVR,LapRo1,Rock}, connections between the Hausdorff dimension of related fractal sets and the {\it topological zeta function} are established and examined. Specifically, in \cite{Gibson}, building upon some examples in \cite{LLVR}, certain collections of fractal strings $\Omega$ with identical sequence of lengths $\mathcal{L}$ are shown to have identical Minkowski dimension but varying Hausdorff dimension (see Fig.~12). The values for the Hausdorff dimension are computed, respectively, as the abscissa of convergence of the topological zeta function.

\ndnt Another interesting extension of our results could lie within the investigation of self-similar measures constructed with weighted IFSs which do not satisfy the OSC. Such an investigation began with an example\footnote{This example was presented by Scott Roby in the poster titled {\it Multifractal analysis of a measure when the open set condition is not satisfied} in the MAA Undergraduate Poster Session at the 2011 Joint Mathematics Meetings in New Orleans.} developed by Scott Roby which was based in no small part upon the results pertaining to second-order identities from \cite{STZ}. However, a full determination of the partition zeta functions in this setting has yet to be discovered.


\as

\textsc{Kate E. Ellis}\\
{\tiny \textsc{Department of Mathematics, California State
University, Stanislaus,\\ Turlock, CA} 95382 \textsc{USA} \par}

\textit{E-mail address:} \textbf{kellis1@csustan.edu}

\s

\textsc{Michel L. Lapidus}\\
{\tiny \textsc{Department of Mathematics, University of
California,\\
Riverside, CA} 92521-0135 \textsc{USA} \par}

\textit{E-mail address:} \textbf{lapidus@math.ucr.edu}

\s

\textsc{Michael C. Mackenzie}\\
{\tiny \textsc{Department of Mathematics, California State
University, Stanislaus,\\ Turlock, CA} 95382 \textsc{USA} \par}

\textit{E-mail address:} \textbf{michael.mackenzie@uconn.edu}

\s

\textsc{John A. Rock}\\
{\tiny \textsc{Department of Mathematics, California State
University, Stanislaus,\\ Turlock, CA} 95382 \textsc{USA} \par}

\textit{E-mail address:} \textbf{jrock@csustan.edu} 


\addcontentsline{toc}{section}{References}

\begin{thebibliography}{99}

\bibitem{AP} M.~Arbeiter and N.~Patzschke, Random self-similar multifractals, \textit{Math. Nachr.} \textbf{181} (1996), 5--42.

\bibitem{BarDem} M.~Barnsley and S.~Demko, Iterated function systems and the global construction of fractals, \textit{Proc. Roy. Soc. London Ser. A}, {\bf 399} (1985), 243--275.

\bibitem{Bes1} A.~S.~Besicovitch, On the sum of digits of real numbers represented in the dyadic system, \textit{Math. Ann.} {\bf 110} (1934) 321--330.

\bibitem{Bes2} A.~S.~Besicovitch, A general form of covering principle and relative differentiation of additive function II, \textit{Proc. Cambridge Philos. Soc.} {\bf 42} (1946) 1--10.

\bibitem{BesTa} A.~S.~Besicovitch and S. J. Taylor, On the complementary intervals of a linear closed set of zero Lebesgue measure, \textit{J. London Math. Soc.} \textbf{29} (1954), 449--459.

\bibitem{BMP} G.~Brown, G.~Michon and J.~Peyri{\`e}re, On the multifractal analysis of measures, \textit{ J. Statist. Phys.} \textbf{66} (1992), 775--790.

\bibitem{CM} R.~Cawley and R.~D.~Mauldin, Multifractal decompositions of Moran fractals, \textit{Adv. Math.} \textbf{92} (1992), 196--236.



\bibitem{DekLi} F.~M.~Dekking and W.~X.~Li, Hausdorff dimensions of subsets of Moran fractals with prescribed group frequency of their codings, \textit{Nonlinearity}, \textbf{16} (2003), 1--13.

\bibitem{EM} G.~A.~Edgar and R.~D.~Mauldin, Multifractal decompositions of digraph recursive fractals, \textit{Proc. London Math. Soc.} \textbf{65} (1992), 604--628.

\bibitem{Egg} H.~G.~Eggleston, The fractional dimension of a set defined by decimal properties, \textit{Quart. J. Math. Oxford Ser.} \textbf{20} (1949), 31--36.

\bibitem{Ellis} R.~S.~Ellis, Large deviations for a general class of random vectors, \textit{Ann. Prob.} \textbf{12} (1984), 1--12.

\bibitem{EssLap} D.~Essouabri and M.~L.~Lapidus, Analytic continuation of a class of multifractal zeta function, preprint (and work in progress), 2010.

\bibitem{EvMan} C.~J.~G.~Evertsz and B.~B.~Mandelbrot, Appendix B of \cite{PeitJS}, pp.~921--953.

\bibitem{Falc} K.~Falconer, \textit{Fractal Geometry -- Mathematical foundations and applications,} 2nd ed., John Wiley, Chichester, 2003.

\bibitem{Fed} J.~Feder, \textit{Fractals}, Plenum Press, New York, 1988.

\bibitem{FengLau} D.-J.~Feng and K.-S.~Lau, Multifractal formalism with weak separation condition, \textit{J. Math. Pures Appl.} \textbf{92} (2000), 407--427.

\bibitem{Gibson} S.~Gibson, Hausdorff dimension and regularity negative infinity, in progress, 2011.

\bibitem{GrMauWi} S.~Graf, R.~D.~Mauldin and S.~C.~Williams, The exact Hausdorff dimension in random recursive constructions, \textit{Mem. Amer. Math. Soc.} No. 381, \textbf{71} (1988), 1--121.

\bibitem{HL} B.~M.~Hambly and M.~L.~Lapidus, Random fractal strings: their zeta functions, complex dimensions and spectral asymptotics, \textit{Trans. Amer. Math. Soc.} No.~1, \textbf{358} (2006), 285--314.

\bibitem{HeLap} C.~Q.~He and M.~L.~Lapidus, Generalized Minkowski content, spectrum of fractal drums, fractal strings and the Riemann zeta-function, \textit{Mem. Amer. Math. Soc.} No. 608, \textbf{127} (1997), 1--97.

\bibitem{HerLap1} H.~Herichi and M.~L.~Lapidus, Invertibility of the spectral operator and a reformulation of the Riemann hypothesis, preprint, 2011.

\bibitem{HerLap2} H.~Herichi and M.~L.~Lapidus, Spectral operator and convergence of the Euler product in the critical strip, in preparation, 2011.

\bibitem{Hut} J.~E.~Hutchinson, Fractals and self-similarity, {\it Indiana Univ. Math. J.} {\bf 30} (1981), 713--747.

\bibitem{Ja1} S.~Jaffard, Multifractal formalism for functions, SIAM \textit{J. Math. Anal.} \textbf{28} (1997), 994--998.

\bibitem{Ja2} S.~Jaffard, Oscillation spaces: properties and applications to fractal and multifractal functions, \textit{J. Math. Phys.} \textbf{38} (1998), 4129--4144.

\bibitem{Ja3} S.~Jaffard, The multifractal nature of L\'{e}vy processes, \textit{Probab. Theory Related Fields} \textbf{114} (1999), 207--227.


\bibitem{JaMey} S.~Jaffard and Y.~Meyer, Wavelet methods for pointwise regularity and local oscilations of functions,
\textit{Mem. Amer. Math. Soc.} No. 587, \textbf{123} (1996), 1--110.

\bibitem{Lap1} M.~L.~Lapidus, Fractal drum, inverse spectral problems for elliptic operators and a partial resolution of the Weyl--Berry conjecture, \textit{Trans. Amer. Math. Soc.} \textbf{325} (1991), 465--529.

\bibitem{Lap2} M.~L.~Lapidus, Spectral and fractal geometry: From the Weyl--Berry conjecture for the vibrations of fractal drums to the Riemann zeta-function, in: \textit{Differential Equations and Mathematical Physics} (C. Bennewitz, ed.), Proc. Fourth UAB Internat. Conf. (Birmingham, March 1990), Academic Press, New York, 1992, pp. 151--182.

\bibitem{Lap3} M.~L.~Lapidus, Vibrations of fractal drums, the Riemann hypothesis, waves in fractal media, and the Weyl--Berry conjecture, in: \textit{Ordinary and Partial Differential Equations} (B. D. Sleeman and R. J. Jarvis, eds.), vol. IV, Proc. Twelfth Internat. Conf. (Dundee, Scotland, UK, June 1992), Pitman Research Notes in Math. Series, vol. 289, Longman Scientific and Technical, London, 1993, pp. 126--209.

\bibitem{Lap4} M.~L.~Lapidus, {\it In Search of the Riemann Zeros: Strings, fractal membranes and noncommutative spacetimes}, Amer. Math. Soc., Providence, RI, 2008.

\bibitem{LLVR} M.~L.~Lapidus, J.~L{\'e}vy-V{\'e}hel and J.~A.~Rock, Fractal strings and multifractal zeta functions, \textit{Letters in Mathematical Physics} No.1, {\bf 88} (2009), 101-129. (Special issue dedicated to the memory of Moshe Flato.) (Also: arXiv:math$\_$ph/0610015v3, 2009.)

\bibitem{LapLu1} M.~L.~Lapidus and H.~Lu, Nonarchimedean Cantor set and string, {\it Journal of Fixed Point Theory and Applications} {\bf 3} (2008), 181--190. (Special issue dedicated to Vladimir Arnold on the occasion of his jubilee; vol. I.) (Also: E-print, Institut des Hautes Etudes Scientifiques, IHES/M/08/29, 2008.)

\bibitem{LapLu2} M.~L.~Lapidus and H.~Lu, Self-similar $p$-adic fractal strings and their complex dimensions, {\it $p$-Adic Numbers, Ultrametric Analysis and Applications}, (Russian Academy of Sciences, Moscow) No.~2, {\bf 1} (2009), 167--180.(Also: E-print, Institut des Hautes Etudes Scientifiques, IHES/M/08/42, 2008.)

\bibitem{LapLu3} M.~L.~Lapidus and H.~Lu, The geometry of $p$-adic fractal strings: A comparative survey, in: {\it Advances in Non-Archimedean Analysis}, Proc. 11th Internat. Conference on ``{\it $p$-Adic Functional Analysis}'' (Clermont-Ferrand, France, July 2010), J.~Araujo, B.~Diarra and A.~Escassut, eds., Contemporary Mathematics, Amer.~Math.~Soc., Providence, R.~I., in press, 2011.

\bibitem{LapLu4} M.~L.~Lapidus, H.~Lu and M.~van Frankenhuijsen, Explicit tube formulas for $p$-adic fractal strings and nonarchimedean self-similar strings, in preparation, 2011.


\bibitem{LapPe1} M.~L.~Lapidus and E.~P.~J.~Pearse, A tube formula for the Koch snowflake curve, with applications to complex dimensions, \textit{J. London Math. Soc.} (2) No.~2, \textbf{74} (2006), 397--414. (Also: E-print arXiv:math-ph/0412029, 2005.)

\bibitem{LapPe3} M.~L.~Lapidus and E.~P.~J.~Pearse, Tube formulas for self-similar factals, in: {\it Analysis on Graphs and its Applications} (P. Exner, {\it et al.}, eds.), Proceedings of Symposia in Pure Mathematics {\bf 77}, Amer. Math. Soc., Providence, RI, 2008, pp.~211--230. (Also: E-print, arXiv:math.DS/0711.0173v1, 2007.)

\bibitem{LapPe2} M.~L.~Lapidus and E.~P.~J.~Pearse, Tube formulas and complex dimensions of self-similar tilings, {\it Acta Applicandae Mathematicae}, \textbf{112} (2010), 91--137. (Also: E-print, arXiv:math.DS/0605527v5, 2010.)

\bibitem{LapPeWin} M.~L.~Lapidus, E.~P.~J.~Pearse and S. Winter, {\it Pointwise tube formulas for fractal sprays and self-similar tilings with arbitrary generators}, {\it Adv. Math.}, in press, 2011. (Also: E-print, arXiv:1006.3807v3 [math.MG], 2011.)

\bibitem{LapPo1} M.~L.~Lapidus and C.~Pomerance, The Riemann zeta-function and the one-dimensional Weyl--Berry conjecture for fractal drums, \textit{Proc. London Math. Soc.} (3) \textbf{66} (1993), 41--69.


\bibitem{LapRo1} M.~L.~Lapidus and J.~A.~Rock, Towards zeta functions and complex dimensions of multifractals, \textit{Complex Variables and Elliptic Equations} No. 6, {\bf 54} (2009), 545--559. (Special issue dedicated to Fractal Analysis.) (Also: E-print, Institut des Hautes Etudes Scientfiques, IHES/M/08/34, 2008.)

\bibitem{LapvF7} M.~L.~Lapidus and M.~van~Frankenhuijsen, Complex dimensions and oscillatory phenomena in fractal geometry and arithmetic, in: {\it Spectral Problems in Geometry and Arithmetic} (T.~Branson, ed.), Contemporary Mathematics {\bf 237}, Amer. Math. Soc., Providence, RI, 1999, pp. 87--105.

\bibitem{LapvF1} M.~L.~Lapidus and M.~van~Frankenhuijsen, \textit{Fractal Geometry and Number Theory: Complex dimensions of fractal strings and zeros of zeta functions,} Birkh\"{a}user, Boston, 2000.


\bibitem{LapvF3} M.~L.~Lapidus and M.~van~Frankenhuijsen, Complex dimensions of self-similar fractal strings and Diophantine approximation, \textit{J. Experimental Mathematics} No.~1, \textbf{42} (2003), 43--69.

\bibitem{LapvF4} M.~L.~Lapidus and M.~van~Frankenhuijsen, \textit{Fractal Geometry, Complex Dimensions and Zeta Functions: Geometry and spectra of fractal strings,} Springer Monographs in Mathematics, Springer-Verlag, New York, 2006.

\bibitem{LapvF5} M.~L.~Lapidus and M.~van~Frankenhuijsen (eds.), \textit{Fractal Geometry and Applications: A Jubilee of Beno\^it Mandelbrot}, Proceedings of Symposia in Pure Mathematics {\bf 72} (Part 1: \textit{Analysis, Number Theory and Dynamical Systems.} Part 2: \textit{Multifractals, Probability and Statistical Mechanics, and Applications}), Amer. Math. Soc., Providence, RI, 2004.

\bibitem{LapvF6} M.~L.~Lapidus and M.~van~Frankenhuijsen, \textit{Fractal Geometry, Complex Dimensions and Zeta Functions: Geometry and spectra of fractal strings,} 2nd rev. and enl. edition, Springer Monographs in Mathematics, Springer-Verlag, New York, to appear in 2011.

\bibitem{Lau} K.~S.~Lau and S.~M.~Ngai, ${L}^q$ spectrum of the {B}ernouilli convolution associated with the golden ration, \textit{Studia Math.} \textbf{131} (1998), 225--251.


\bibitem{LVM} J.~L{\'e}vy-V{\'e}hel and F.~Mendivil, Multifractal and higher-dimensional zeta functions, {\it Nonlinearity}, No.~1, {\bf 24} (2011), 259--276.



\bibitem{LVT} J.~L{\'e}vy-V{\'e}hel and C.~Tricot, On various multifractal sprectra, in: \textit{Fractal Geometry and Stochastics III} (C. Bandt, U. Mosco and M. Z\"{a}hle, eds.), Birkh\"{a}user, Basel, 2004, pp. 23--42.

\bibitem{LVV} J.~L{\'e}vy-V{\'e}hel and R.~Vojak, Multifractal analysis of Choquet capacities, \textit{Adv. Appl. Math.} \textbf{20} (1998), 1--43.

\bibitem{LOW} W.~Li, L.~Olsen and Z.~Wen, Hausdorff and packing dimensions of subsets of Moran fractals with prescribed mixed group frequency in their codings, \textit{Aequationes Math.} \textbf{77} (2009), 171--185.


\bibitem{BM} B.~B.~Mandelbrot, Intermittent turbulence in self-similar cascades: divergence of hight moments and dimension of the carrier, \textit{J. Fluid. Mech.} \textbf{62} (1974), 331--358.

\bibitem{Man} B.~B.~Mandelbrot, \textit{Multifractals and $1/f$ Noise,} Springer-Verlag, New York, 1999.

\bibitem{MinYa} W.~Min and W.~Yahao, Dimensions of modified Besicovitch subsets of Moran fractal, \textit{Chaos, Solitons and Fractals} {\bf 42} (2009), 2779--2785.

\bibitem{Mor} P.~A.~P.~Moran, Additive functions of intervals and Hausdorff measure, {\it Math. Proc. Cambridge Philos. Soc.} {\bf 42} (1946), 15--23.


\bibitem{Ol2} L.~Olsen, A multifractal formalism, \textit{Adv. Math.} \textbf{116} (1996), 82--196.

\bibitem{Ol3} L.~Olsen, Multifractal geometry, in: \textit{Fractal Geometry and Stochastics II} (Greifswald/Koserow, 1998), \textit{Prog. Prob.} \textbf{46} (2000), Birkh\"{a}user, Basel, pp. 3--37.

\bibitem{Ol4} L.~Olsen, A lower bound for the symbolic multifractal spectrum of a self-similar multifractal with arbitrary overlaps, \textit{Math. Nachr.} No. 10, \textbf{282} (2009), 1461--1477.

\bibitem{PF} G.~Parisi and U.~Frisch,  Fully developed turbulence and intermittency inturbulence, and predictability in geophysical fluid dynamics and climate dynamics, in: \textit{International School of ``Enrico Fermi'', Course 88} (M. Ghil, ed.), North-Holland, Amsterdam, 1985, pp. 84--88.

\bibitem{PeitJS} H.-O.~Peitgen, H.~J\"urgens and D.~Saupe, \textit{Chaos and Fractals: New frontiers of science}, Springer-Verlag, Berlin, 1992. (See esp. Appendix B by C.~J.~G.~Evertsz and B.~B.~Mandelbrot, pp. 921--953.)


\bibitem{Rock} J.~A.~Rock, Zeta Functions, Complex Dimensions of Fractal Strings and Multifractal Analysis of Mass Distributions, \textit{Ph.D. Dissertation}, University of California, Riverside, 2007.

\bibitem{Sch} M.~R.~Schroeder, \textit{Fractals, Chaos, Power Laws: Minutes from an infinite paradise}, W.~H. Freeman, New York, 1991.

\bibitem{STZ} R.~S.~Strichartz, A.~ Taylor, and T.~ Zhang, Densities of Self-Similar Measures on the Line, \textit{J.~Experimental Mathematics} No. 2, {\bf 4} (1995).

\end{thebibliography}
\end{document}